\documentclass[onecolumn,12pt]{IEEEtran}
\usepackage[mathscr]{eucal}
\usepackage[cmex10]{amsmath}
\usepackage{epsfig,epsf,psfrag}
\usepackage{amssymb,amsmath,amsthm,amsfonts,latexsym}
\usepackage{amsmath,graphicx,bm,xcolor,url,overpic}
\usepackage{array}
\usepackage{verbatim}
\usepackage{bm}
\usepackage{algorithmic}
\usepackage{algorithm}
\usepackage{verbatim}
\usepackage{textcomp}
\usepackage{mathrsfs}
\usepackage{epstopdf}

\newcommand{\openone}{\leavevmode\hbox{\small1\normalsize\kern-.33em1}}

\catcode`~=11 \def\UrlSpecials{\do\~{\kern -.15em\lower .7ex\hbox{~}\kern .04em}} \catcode`~=13 

\allowdisplaybreaks[4]

\newcommand{\nn}{\nonumber}

\newcommand{\calA}{\mathcal{A}}
\newcommand{\calB}{\mathcal{B}}
\newcommand{\calC}{\mathcal{C}}
\newcommand{\calD}{\mathcal{D}}

\newcommand{\calF}{\mathcal{F}}

\newcommand{\calH}{\mathcal{H}}

\newcommand{\calM}{\mathcal{M}}

\newcommand{\calP}{\mathcal{P}}

\newcommand{\calT}{\mathcal{T}}

\newcommand{\calW}{\mathcal{W}}
\newcommand{\calX}{\mathcal{X}}
\newcommand{\calY}{\mathcal{Y}}

\newcommand{\bx}{\mathbf{x}}
\newcommand{\bX}{\mathbf{X}}
\newcommand{\by}{\mathbf{y}}
\newcommand{\bY}{\mathbf{Y}}

\newcommand{\rmc}{\mathrm{c}}

\newcommand{\rmd}{\mathrm{d}}

\newcommand{\rmf}{\mathrm{f}}

\newcommand{\rmG}{\mathrm{G}}

\newcommand{\rmH}{\mathrm{H}}

\newcommand{\rmr}{\mathrm{r}}

\newcommand{\rms}{\mathrm{s}}
\newcommand{\rmS}{\mathrm{S}}

\newcommand{\bbE}{\mathsf{E}}

\newcommand{\bbN}{\mathbb{N}}

\newcommand{\bbP}{\mathbb{P}}

\newcommand{\bbR}{\mathbb{R}}

\DeclareMathAlphabet{\mathbsf}{OT1}{cmss}{bx}{n}
\DeclareMathAlphabet{\mathssf}{OT1}{cmss}{m}{sl}

\DeclareSymbolFont{bsfletters}{OT1}{cmss}{bx}{n}  
\DeclareSymbolFont{ssfletters}{OT1}{cmss}{m}{n}
\DeclareMathSymbol{\bsfGamma}{0}{bsfletters}{'000}
\DeclareMathSymbol{\ssfGamma}{0}{ssfletters}{'000}
\DeclareMathSymbol{\bsfDelta}{0}{bsfletters}{'001}
\DeclareMathSymbol{\ssfDelta}{0}{ssfletters}{'001}
\DeclareMathSymbol{\bsfTheta}{0}{bsfletters}{'002}
\DeclareMathSymbol{\ssfTheta}{0}{ssfletters}{'002}
\DeclareMathSymbol{\bsfLambda}{0}{bsfletters}{'003}
\DeclareMathSymbol{\ssfLambda}{0}{ssfletters}{'003}
\DeclareMathSymbol{\bsfXi}{0}{bsfletters}{'004}
\DeclareMathSymbol{\ssfXi}{0}{ssfletters}{'004}
\DeclareMathSymbol{\bsfPi}{0}{bsfletters}{'005}
\DeclareMathSymbol{\ssfPi}{0}{ssfletters}{'005}
\DeclareMathSymbol{\bsfSigma}{0}{bsfletters}{'006}
\DeclareMathSymbol{\ssfSigma}{0}{ssfletters}{'006}
\DeclareMathSymbol{\bsfUpsilon}{0}{bsfletters}{'007}
\DeclareMathSymbol{\ssfUpsilon}{0}{ssfletters}{'007}
\DeclareMathSymbol{\bsfPhi}{0}{bsfletters}{'010}
\DeclareMathSymbol{\ssfPhi}{0}{ssfletters}{'010}
\DeclareMathSymbol{\bsfPsi}{0}{bsfletters}{'011}
\DeclareMathSymbol{\ssfPsi}{0}{ssfletters}{'011}
\DeclareMathSymbol{\bsfOmega}{0}{bsfletters}{'012}
\DeclareMathSymbol{\ssfOmega}{0}{ssfletters}{'012}

\newcommand{\tilP}{\tilde{P}}

\newcommand{\tilbbP}{\tilde{\bbP}}

\newcommand{\tilQ}{\tilde{Q}}

\newcommand{\hatT}{\hat{T}}

\newcommand{\bart}{\bar{t}}

\newcommand{\barh}{\bar{h}}

\newcommand{\tPhi}{\tilde{\Phi}}

\DeclareMathOperator*{\argmin}{arg\,min}

\newtheorem{theorem}{Theorem} 
\newtheorem{lemma}{Lemma}

\usepackage{subfigure,epstopdf,bbm} 
\usepackage[utf8]{inputenc} 
\usepackage[T1]{fontenc}    
\usepackage{url,bbm}            
\usepackage{booktabs}       
\usepackage{amsfonts}       
\usepackage{nicefrac}       
\usepackage{microtype}      
\usepackage{cite}
\usepackage{subfigure,transparent,color,graphicx} 

\usepackage[ colorlinks = true,
linkcolor = blue,
urlcolor  = blue,
citecolor = red,
anchorcolor = green]{hyperref}
\allowdisplaybreaks[2]
\begin{document}
\title{Large Deviations for Sequential Tests of Statistical Sequence Matching}
\author{Lin Zhou, Qianyun Wang, Yun Wei and Jingjing Wang

\thanks{L. Zhou, Q. Wang and J. Wang are with the School of Cyber Science and Technology, Beihang University (Emails: \{lzhou, wangqianyun, drwangjj\}@buaa.edu.cn).}
\thanks{Y. Wei is with the Department of Mathematical Sciences, University of Texas at Dallas (Email: yun.wei@utdallas.edu).}
}

\maketitle

\begin{abstract}
We revisit the problem of statistical sequence matching initiated by Unnikrishnan (TIT 2015) and derive theoretical performance guarantees for sequential tests that have bounded expected stopping times. Specifically, in this problem, one is given two databases of sequences and the task is to identify all matched pairs of sequences. In each database, each sequence is generated i.i.d. from a distinct distribution and a pair of sequences is said matched if they are generated from the same distribution. The generating distribution of each sequence is \emph{unknown}. We first consider the case where the number of matches is known and derive the exact exponential decay rate of the mismatch (error) probability, a.k.a. the mismatch exponent, under each hypothesis for optimal sequential tests. Our results reveal the benefit of sequentiality by showing that optimal sequential tests have larger mismatch exponent than fixed-length tests by Zhou \emph{et al.} (TIT 2024). Subsequently, we generalize our achievability result to the case of unknown number of matches. In this case, two additional error probabilities arise: false alarm and false reject probabilities. We propose a corresponding sequential test, show that the test has bounded expected stopping time under certain conditions, and characterize the tradeoff among the exponential decay rates of three error probabilities. Furthermore, we reveal the benefit of sequentiality over the two-step fixed-length test by Zhou \emph{et al.} (TIT 2024) and propose an one-step fixed-length test that has no worse performance than the fixed-length test  by Zhou \emph{et al.} (TIT 2024). When specialized to the case where either database contains a single sequence, our results specialize to large deviations of sequential tests for statistical classification, the binary case of which was recently studied by Hsu, Li and Wang (ITW 2022).
\end{abstract}


\section{Introduction}
Motivated by practical applications including image classification and junk mail identification, Gutman~\cite{gutman1989asymptotically} proposed statistical classification where training sequences are provided to help decide the generating distribution of a testing sequence. Specifically, in the simplest binary case, one is given three i.i.d. sequences: a testing sequence and two training sequences. Each training sequence is generated i.i.d. from a distinct but unknown distribution. The task is to infer the generating distribution of the testing sequence by checking whether the testing sequence is generated from the same distribution as one of the training sequences. Statistical classification generalizes hypothesis testing to the case of unknown generating distributions and specializes to hypothesis testing when the length of each training sequence is infinite so that the generating distributions can be estimated accurately. When the lengths of testing and training sequences are fixed, Gutman~\cite{gutman1989asymptotically} proposed a non-parametric test, analyzed its asymptotic performance and demonstrated that the test is asymptotically optimal in the generalized Neyman-Pearson sense when the lengths of all sequences tend to infinity. Zhou, Tan and Motani~\cite{zhou2018binary} refined Gutman's result by characterizing the second-order asymptotic bounds that approximate the non-asymptotic performance of optimal tests when all sequences have finite lengths.

Analogous to hypothesis testing where sequentiality improves the performance of optimal tests~\cite{wald1945sequential,wald1948optimum}, the superior performance of sequential test for statistical classification was first revealed by Haghifam, Tan, and Khisti~\cite{mahdi2021sequential} in a semi-sequential setting, where two training sequences of fixed-lengths are available while the testing sequence arrives in a sequential manner. It was shown in~\cite[Theorem 3]{mahdi2021sequential} that the Bayesian exponent of the sequential test is strictly larger than that of Gutman's fixed-length test. Subsequently, Hsu, Li and Wang~\cite{Ihwang2022sequential} studied a fully sequential setting where the testing sequence and all training sequences arrive sequentially. In particular, the authors studies two universal settings: expected stopping time universality and error probability universality. In the first universal setting, the expected stopping time of the test is bounded regardless of the generating distributions of training sequences while in the second universal setting, the error probability of the test is bounded regardless of the generating distributions. Furthermore, the benefit of sequentiality was proved for optimal tests under the first universal setting~\cite[Remark 4]{Ihwang2022sequential}. Very recently, Li and Wang~\cite{li2025} unified the studies of sequential settings and revealed the benefit of sequentiality in all three cases where either the testing sequence arrives sequentially or training sequence arrives sequentially, under the expected stopping time universality constraint.

In another generalization of Gutman's study~\cite{gutman1989asymptotically}, motivated by practical applications including data de-anonymization~\cite{Unnikrishnan2013}, Unnikrishnan~\cite{unnikrishnan2015asymptotically} proposed the problem of statistical sequence matching. Specifically, one is given two databases of sequences, where in each database, each sequence is generated i.i.d. from a distinct distribution. A pair of sequences across two databases is said matched if they are generated from the same distribution. The task of statistical sequence matching is to identify all matched pairs of sequences. Note that when either database contains a single sequence, statistical sequence matching specializes to statistical classification~\cite{gutman1989asymptotically}. When the number of matches is known, Unnikrishnan~\cite{unnikrishnan2015asymptotically} proposed an asymptotically optimal fixed-length test in the spirit of Gutman. Very recently, the results of Unnikrishnan were generalized to non-asymptotic analyses and to unknown number of matches by Zhou \emph{et al.}~\cite{zhou2024tit}.

Although the fixed-length tests have been well studied for statistical sequence matching, sequential tests have not been addressed and the potential benefit of sequentiality remains to be explored. In this paper, we address this research gap by deriving exact large deviations for optimal sequential tests when the number of matches is known and by deriving achievability results of large deviations when the number of matches is unknown. Our contributions are summarized as follows.

\subsection{Main Contributions}
We first consider the case where the number of matches is known. Our contributions are four fold. Firstly, we propose a non-parametric sequential test, show that the test has bounded expected stopping time under any tuple of generating distributions, and derive the exponential decay rate of the mismatch probability, a.k.a., the mismatch exponent. Secondly, we prove that our proposed sequential test is optimal by having the largest mismatch exponent among all sequential tests that have bounded expected stopping times. Thirdly, when specialized to statistical classification, our results bound the exact large deviations for sequential classification and specialize to \cite[Theorem 1]{Ihwang2022sequential} for binary classification. Finally, by generalizing the test and analysis for \cite[Theorem 2]{zhou2024tit}, we demonstrate the benefit of sequentiality by showing that i) our optimal sequential test has larger mismatch exponent than the fixed-length test using the minimal scoring function decision rule and ii) the fixed-length test in \cite[Theorem 2]{zhou2024tit} has no better performance than the fixed-length test applying the minimal scoring function decision rule.

We next generalize our achievability results to the case of unknown number of matches. In this case, the number of matches can be zero. Correspondingly, two additional error events occur: false alarm and false reject. Our contributions for this case are five fold. Firstly, we propose a non-parametric sequential test and demonstrate its asymptotic intuition using the weak law of large numbers, which sets solid foundation for theoretical analyses. Secondly, we show that our sequential test has bounded expected stopping times under mild conditions on its parameters and we lower bound the exponential decay rates of all three error probabilities. Thirdly, when specialized to statistical classification, we lower bound the achievable exponents for sequential tests of statistical classification. Our specialization complements previous studies in~\cite{gutman1989asymptotically,Ihwang2022sequential} by allowing the testing sequence to be generated from a distribution that is different from the generating distribution of all training sequences. Fourthly, comparing with the corresponding result for fixed-length test~\cite[Theorem 4]{zhou2024tit}, we demonstrate that our sequential test has larger Bayesian exponent when there exists matched pair of sequences and achieves the same false alarm exponent when there is no matched pair of sequences. Finally, we propose a one-step fixed-length test that refines the fixed-length test of~\cite[Theorem 4]{zhou2024tit}: our proposed fixed-length test achieves the same false alarm exponent under the null hypothesis, our proposed fixed-length test achieves the same Bayesian exponent under each non-null hypothesis and in particular, ur proposed fixed-length test proceeds in one phase, which does not need to estimate the number of matches before identifying the set of matched sequences.

\subsection{Organization for the Rest of the Paper}
The rest of the paper is organized as follows. In Section \ref{sec:pf}, we set up the notation and formulate the problem of statistical sequence matching for both cases of known and unknown number of matches. In Section \ref{sec:results}, we present test design, theoretical results and discussions for the case of known number of matches. Subsequently, the results for unknown number of matches are presented in Section \ref{sec:results:unknown}. The proofs for our results are presented in Sections \ref{proof:known} and \ref{proof:unknown}. Finally, we summarize our contributions and discuss future research directions in Section \ref{sec:conclusion}.

\section{Problem Formulation and Existing Results}
\label{sec:pf}

\subsection*{Notation}
Random variables and their realizations are in upper case (e.g.,  $X$) and lower case (e.g.,  $x$), respectively. All sets are denoted in calligraphic font (e.g.,  $\mathcal{X}$). We use $\bbR$, $\bbR_+$, and $\bbN$ to denote the set of real numbers, non-negative real numbers, and natural numbers, respectively. Given any integer $a\in\bbN$ such that $a\geq 1$, we use $[a]$ to denote the collection of natural numbers between $1$ and $a$ and use $\bbN_a$ to denote the collection of natural numbers that are greater than or equal to $a$. We use superscripts to denote the length of vectors, e.g., $X^n:=(X_1,\ldots,X_n)$. All logarithms are base $e$. The set of all probability distributions on a finite set $\calX$ is denoted as $\calP(\calX)$. Notation concerning the method of types follows~\cite{TanBook,ZhouBook}. Specifically, given a vector $x^n = (x_1,x_2,\ldots,x_n) \in\calX^n$, the {\em type} or {\em empirical distribution} is denoted as $\hatT_{x^n}(a)=\frac{1}{n}\sum_{i=1}^n \mathbbm{1}\{x_i=a\},a\in\calX$. The set of types formed from length-$n$ sequences with alphabet $\calX$ is denoted as $\calP_{n}(\calX)$. Given $P\in\calP_{n}(\calX)$, the set of all sequences of length $n$ with type $P$, the {\em type class}, is denoted as $\calT^n_P$.

\subsection{Case of Known Number of Matches}
\label{sec:pf:known}

\subsubsection{System Model}
We first consider the case where the number of matches is known. Fix integers $(M_1,M_2,K)\in\bbN^3$ such that $M_1\geq M_2\geq K$ and fix two positive real numbers $(\alpha,\beta)\in\bbR_+$. For each integer $n\in\bbN$, let $\xi_n:=\lceil \alpha n\rceil $ and $\chi_n:=\lceil \beta n\rceil$. Furthermore, let $\bX^{\xi_n}:=\{X_1^{\xi_n},\ldots,X_{M_1}^{\xi_n}\}$ denote the first database with $M_1$ sequences, where for each $i\in[M_1]$, $X_i^{\xi_n}=(X_{i,1},\ldots,X_{i,\xi_n})$ is generated i.i.d. from an unknown distribution $P_i$ defined on the finite alphabet $\calX$. Let $\bY^{\chi_n}:=\{Y_1^{\chi_n},\ldots,Y_{M_2}^{\chi_n}\}$ be the second database of $M_2$ sequences, where for each $i\in[M_2]$, $Y_i^{\chi_n}=(Y_{i,1},\ldots,Y_{i,\chi_n})$ is generated i.i.d. from an unknown distribution $Q_i$ defined on $\calX$. A pair of sequences is said be matched if they are generated from the same distribution. Following~\cite{unnikrishnan2015asymptotically}, we assume that each sequence in each database is generated by a distinct distribution and there are $K$ matched pairs of sequences. 

Note that there are in total $T_K:={M_1\choose K}{M_2\choose K}K!$ possibilities of $K$-matches between the two databases. To represent each possibility (hypothesis) explicitly, we need the following definitions. Given any $l\in[T_K]$, under hypothesis $\rmH_l^K$, let $\calM_l^K\in([M_1]\times[M_2])^K$ collect the indices of matched pairs such that for any $(i,j)\in\calM_l^K$, the $i$-th sequence of the first database is mapped to the $j$-th sequence in the second database. Furthermore, let $\calM^K$ be the collection of all $T_K$ possibilities. Given any $t\in[T_K]$, define two sets
\begin{align}
\calC_t^K&:=\big\{i\in[M_1]:~\exists~j\in[M_2],~(i,j)\in\calM_t^K\big\}\label{def:calc},\\
\calD_t^K&:=\big\{j\in[M_2]:~\exists~i\in[M_1],~(i,j)\in\calM_t^K\big\}\label{def:cald}.
\end{align}
Under hypothesis $\rmH_l^K$, $\calC_t^K$ collects the indices of matched sequences in the first database while $\calD_t^K$ collects the indices of matched sequences in the second database. To illustrate the above definitions, we provide two examples:
\begin{itemize}
\item When $M_1=3$, $M_2=2$ and $K=1$. In this case, $T_K={3 \choose 1}{2\choose 1} 1!=6$ and $\calM^K$ consists of
\begin{align}
\begin{array}{ccc}
\calM_1^K=\{(1,1)\} & \calM_2^K=\{(2,1)\} & \calM_3^K=\{(3,1)\}\\
\calM_4^K=\{(1,2)\} &\calM_5^K= \{(2,2)\} & \calM_6^K=\{(3,2)\}.
\end{array}
\end{align}
When $l=2$, $\calM_l^K=\{(2,1)\}$, $\calC_l^K=\{2\}$, $\calD_l^K=\{1\}$, which means that $P_2=Q_1$. In other words, the second sequence of the first database is matched to the first sequence of the second database.

\item When $M_1=4$, $M_2=2$ and $K=2$. In this case, $T_K={4\choose 2}{2\choose 2}2!=12$ and $\calM^K$ consists of
\begin{align}
\begin{array}{cccc}
\calM_1^K=\{(1,1),(2,2)\} & \calM_2^K=\{(1,2),(2,1)\} & \calM_3^K=\{(1,1),(3,2)\} & \calM_4^K=\{(1,2),(3,1)\}\\
\calM_5^K=\{(1,1),(4,2)\} & \calM_6^K=\{(1,2),(4,1)\} & \calM_7^K=\{(2,1),(3,2)\} & \calM_8^K=\{(2,2),(3,1)\}\\
\calM_9^K=\{(2,1),(4,2)\} & \calM_{10}^K=\{(2,2),(3,1)\} & \calM_{11}^K=\{(3,1),(4,2)\} & \calM_{12}^K=\{(3,2),(4,1)\}.
\end{array}
\end{align}
When $l=3$, $\calM_l^K=\{(1,1),(3,2)\}$, $\calC_l^K=\{1,3\}$, $\calD_l^K=\{1,2\}$, which means that $P_1=Q_1$ and $P_3=Q_2$. In other words, the first sequence of the first database is matched to the first sequence of the second database while the third sequence of the first database is matched to the second sequence of the second database.
\end{itemize}

The task of sequential sequence matching is to design a test $\Phi=(\tau,\phi_\tau)$ with a random stopping time $\tau$ and a decision rule $\phi_\tau:\calX^{M_1\tau}\times\calY^{M_2\tau}\to \{\rmH_l^K\}_{l\in[T_K]}$. The stopping time $\tau$ is a function of the filtration $\{\calF_n\}_{n\in\bbN}$, where $\calF_n=\sigma\{\bX^{\xi_n},\bY^{\chi_n}\}$. As argued by Unnikrishnan~\cite{unnikrishnan2015asymptotically}, if the distinct distribution assumption is removed, the problem reduces to repeated version of the statistical classification problem~\cite{gutman1989asymptotically,zhou2018binary} in the sequential setting~\cite{Ihwang2022sequential,zhou2022twophase}. Furthermore, if $K=M_2=1$, this problem is exactly the $M_1$-ary classification problem. Therefore, the sequential sequence matching problem strictly generalizes the sequential classification problem.

\subsubsection{Performance Metric}
Fix any $l\in[T_K]$. Define the following set of generating distributions
\begin{align}
\calP_l^K:=
\big\{(\tilP^{M_1},\tilQ^{M_2})\in\calP(\calX)^{M_1+M_2}:~\tilP_i=\tilQ_j~\mathrm{iff}~(i,j)\in\calM_l^K\big\}\label{def:calp:lk}.
\end{align}
Note that $\calP_l^K$ is the set of all possible tuples of generating distributions under hypothesis $\rmH_l^K$. To evaluate the performance of a test $\Phi$, under hypothesis $\rmH_l^K$ and any tuple of generating distributions $(P^{M_1},Q^{M_2})\in\calP_l^K$, the following mismatch probability is considered:
\begin{align}
\beta(\Phi|P^{M_1},Q^{M_2})&:=\bbP_l^K\big\{\phi_\tau(\bX^{\xi_\tau},\bY^{\chi_\tau})\neq \rmH_l^K\big\}\label{def:mismatch},
\end{align}
where $\bbP_l^K$ denotes the joint distribution of all sequences. The mismatch probability corresponds to the probability that an incorrect $K$-match is decided.

Since the random stopping time varies, it would be desirable to constrain its expected value. A test $\Phi=(\tau,\phi_\tau)$ is called an expected stopping time universality test if there exists an integer $N\in\bbN$ such that  the expected stopping time under any tuple of generating distributions is bounded by $N$, i.e.,
\begin{align}
\label{est:uni}
\max_{l\in[T_K]}\max_{(P^{M_1},Q^{M_2})\in\calP_l^K}\bbE_{\bbP_l^K}[\tau]\leq N.
\end{align}
The notion of stopping time universality test was proposed by Hsu, Li and Wang in their study of sequential tests for binary classification~\cite{Ihwang2022sequential}. 

A special case of the expected stopping time universality test is a fixed-length test, when $\tau$ is set to be $N$. This case was previously studied by Unnikrishnan~\cite{unnikrishnan2015asymptotically} and by Zhou \emph{et al.}~\cite{zhou2024tit}. In these studies, an additional null hypothesis ($\rmH_\rmr$ in the next subsection) was introduced to account for the case where a reliable decision could not be made. With this additional null hypothesis, the large deviations performance of the optimal test in the generalized Neyman-Pearson sense was characterized.

\subsection{Case of Unknown Number of Matches}
\label{sec:pf:unknown}

A more practical setting is where the number of matches $K$ is \emph{unknown} a priori. In this case, the number of matches needs to be estimated and all pairs of matched sequences should be identified. To account for the possibility that no match between two databases exists, we define the null hypothesis, denoted as the reject hypothesis $\rmH_\rmr$, which corresponds to $K=0$. For each $K\in[M_2]$, we use $\calH_K$ to denote the set of all $T_K$ hypotheses when the number of matches is $K$. Thus, when the number of matches is unknown, the total number of hypotheses increases to $T+1$ where $T:=\sum_{K\in[M_2]}T_K$. 

Correspondingly, our task is to design a test $\Phi=(\tau,\phi_\tau)$ with a random stopping time $\tau$ and a test 
$\phi_\tau: \calX^{M_1\xi_\tau}\times\calY^{M_2\chi_\tau}\to\{\{\calH_K\}_{K\in[M_2]},\rmH_{\rmr}\}$ to classify among the following hypotheses:
\begin{itemize}
\item $\rmH_l^K\in\calH_K$ where $K\in[M_2]$ and 
$l\in[T_K]$: for each $(i,j)\in\calM_l^K$, the $i$-th sequence from the first database is mapped to the $j$-th sequence of the second database.
\item $\rmH_\rmr$: there is no matched pair of sequences in the two databases.
\end{itemize} 

To evaluate the performance of a test, fix any $K\in[M_2]$ and $l\in[T_K]$, under the non-null hypothesis $\rmH_l^K$ and generating distributions $(P^{M_1},Q^{M_2})\in\calP_l^K$, we consider the following modified mismatch probability and the false reject probability:
\begin{align}
\bar{\beta}(\Phi|P^{M_1},Q^{M_2})&:=\bbP_l^K\big\{\phi_\tau(\bX^{\xi_\tau},\bY^{\chi_\tau})\notin\{\rmH_l^K,\rmH_\rmr\}\big\},\label{def:mismatch:unknown}\\
\zeta(\Phi|P^{M_1},Q^{M_2})&:=\bbP_l^K\big\{\phi_\tau(\bX^{\xi_\tau},\bY^{\chi_\tau})=\rmH_\rmr\big\}\label{def:freject}.
\end{align}
Note that $\bar{\beta}(\Phi|P^{M_1},Q^{M_2})$ corresponds to the probability that an incorrect $K$-match is decided under hypothesis $\rmH_l^K$, while $\zeta(\Phi|P^{M_1},Q^{M_2})$ corresponds to the probability that a no-match decision is output under hypothesis $\rmH_l^K$ when there exist matched pairs of sequences.

Analogously to \eqref{def:calp:lk}, define the following set of possible distributions under the null hypothesis:
\begin{align}
\calP_0:=\big\{(\tilP^{M_1},\tilQ^{M_2})\in\calP(\calX)^{M_1+M_2}:~\forall~(i,j)\in[M_1]\times[M_2],~\tilP_i\neq \tilQ_j\big\}
\label{def:calP:r}.
\end{align}
Under the null hypothesis $\rmH_\rmr$, for any tuple of generating distributions $(P^{M_1},Q^{M_2})\in\calP_0$, we also need the following false alarm probability:
\begin{align}
\eta(\Phi|P^{M_1},Q^{M_2}):=\bbP_\rmr\big\{\phi_\tau(\bX^{\xi_\tau},\bY^{\chi_\tau})\neq \rmH_\rmr\big\}\label{def:etar},
\end{align}
where $\bbP_\rmr$ denotes the joint distribution of all sequences under the null hypothesis. Note that $\eta(\Phi|P^{M_1},Q^{M_2})$ quantifies the probability that the test declares that there exists some matched pair of sequences when there is none.

When the number of matches is unknown, a test $\Phi=(\tau,\phi_\tau)$ is called an expected stopping time universality test if \eqref{est:uni} is satisfied and there exists an integer $N\in\bbN$ such that 
\begin{align}
\max_{(P^{M_1},Q^{M_2})\in\calP_0}\bbE_{\bbP_\rmr}[\tau]\leq N.
\end{align}
In this case, the crux is to study the tradeoff among the probabilities of mismatch, false reject and false alarm for sequential tests that have bounded expected stopping times under each hypothesis.

\section{Main Results for the Case of Known Number of Matches}
\label{sec:results}

\subsection{Preliminaries}
Fix any pair of distributions $(P,Q)\in\calP(\calX)^2$ with full support. The KL divergence is defined as
\begin{align}
D(P\|Q):=\sum_{x\in\calX}P(x)\log\frac{P(x)}{Q(x)}.
\end{align}
Given any positive real number $\alpha\in\bbR_+$, the R\'enyi Divergence of order $\alpha$~\cite[Eq. (1)]{van2014renyi} is defined as
\begin{align}
D_\alpha(P||Q):=
\left\{
\begin{array}{ll}
D(P\|Q)&\mathrm{if~}\alpha=1,\\
\frac{1}{\alpha-1}\log\sum_{x\in\calX}P(x)^{\alpha}Q(x)^{1-\alpha}&\mathrm{otherwise}.
\end{array}
\right.
\label{def:renyi}
\end{align}
The R\'{e}nyi Divergence has the following variational form~\cite[Eq. (7)]{Ihwang2022sequential}:
\begin{align}
\label{renyi:variational}
D_{\frac{\alpha}{1+\alpha}}(P||Q):=\min_{V\in\calP(\calX)}\big(\alpha D(V||P)+D(V||Q)\big).
\end{align}

Fix any two positive real numbers $(\alpha,\beta)\in\bbR_+^2$. Define the following linear combination of distributions $(P,Q)$:
\begin{align}
R_{\alpha,\beta}^{P,Q}:=\frac{\alpha P+\beta Q}{\alpha+\beta}\label{def:Rab},
\end{align}
and define the following linear combination of KL divergence:
\begin{align}
\mathrm{GJS}(P,Q,\alpha,\beta):=\alpha D\big(P\|R_{\alpha,\beta}^{P,Q}\big)+\beta D\big(Q\|R_{\alpha,\beta}^{P,Q}\big)\label{def:GJS}.
\end{align}
Note that $\mathrm{GJS}(P,Q,\alpha,\beta)$ measures the distance between $P$ and $Q$, which equals zero if and only if $P=Q$.  When $\beta=1$, $\mathrm{GJS}(P,Q,\alpha,\beta)$ specializes to the generalized Jensen-Shannon divergence~\cite[Eq. (2.3)]{zhou2018binary} for classification and further specializes to twice of Jensen-Shannon divergence~\cite[Eq. (4.1)]{lin1991divergence} when $\alpha=1$. Similar definition has also been used in other statistical inference problems for test design and theoretical benchmark presentation, e.g.,~\cite{
li2014,unnikrishnan2015asymptotically,zhou2018binary,he2019distributed,hsu2020binary,mahdi2021sequential,zhou2022second,zhou2022twophase}.

Given any two sets of distributions $P^{M_1}=(P_1,\ldots,P_{M_1})\in\calP(\calX)^{M_1}$ and $Q^{M_2}=(Q_1,\ldots,Q_{M_2})\in\calP(\calX)^{M_2}$, for each $t\in[T_K]$, let
\begin{align}
\rmG_t^K(P^{M_1},Q^{M_2},\alpha,\beta)
&:=\sum_{(i,j)\in\calM_t^K}\mathrm{GJS}(P_i,Q_j,\alpha,\beta)
\label{def:Gl}.
\end{align}
Fix any integer $n\in\bbN$ and any realizations of two databases $\bx^{\xi_n}=\{x_1^{\xi_n},\ldots,x_{M_1}^{\xi_n}\}$ and $\by^{\chi_n}=\{y_1^{\chi_n},\ldots,y_{M_2}^{\chi_n}\}$. Let $\hatT_{\bx^{\xi_n}}:=(\hatT_{x_1^{\xi_n}},\ldots,\hatT_{x_{M_1}^{\xi_n}})$ and let $\hatT_{\by^{\chi_n}}:=(\hatT_{y_1^{\chi_n}},\ldots,\hatT_{y_{M_2}^{\chi_n}})$ be the collection of empirical distributions. For each $t\in[T_K]$, define the scoring function 
\begin{align}
\rmS_t^K(\bx^{\xi_n},\by^{\chi_n})&:=\rmG_t^K(\hatT_{\bx^{\xi_n}},\hatT_{\by^{\chi_n}},\alpha,\beta)\label{def:Sl}.
\end{align}

\subsection{Sequential Test and Asymptotic Intuition}
\label{asymp:intuition}

Our sequential test $\Phi=(\tau,\phi_\tau)$
consists of a random stopping time $\tau$ and a corresponding decision rule $\phi_\tau$ to identify the pairs of matched sequences. For each $n\in\bbN$, define a function $f(n)$ such that
\begin{align}
f(n):=\frac{(K+1)|\calX|\log (n\alpha+2)+K|\calX|\log(n\beta+2)}{n}\label{def:fn}.
\end{align}
Fix an integer $N\in\bbN$. The random stopping time is chosen such that
\begin{align}
\tau:=\inf\Big\{n\in\bbN_{N-1}:~\exists~t\in[T_K],~\rmS_t^K(\bX^{\xi_n},\bY^{\chi_n})\leq f(n)\Big\}\label{def:tau}.
\end{align}
At the random stopping time $\tau$, the minimal scoring function decision rule is used, i.e., $\phi_\tau(\bX^{\xi_n},\bY^{\chi_n})=\rmH_l^K$ if
\begin{align}
l=\argmin_{t\in[T_K]}\rmS_t^K(\bX^{\xi_n},\bY^{\chi_n})\label{test:tau}.
\end{align}

We first explain why the above sequential test works asymptotically using the weak law of large numbers. Fix any $l\in[T_K]$ and consider any tuple of generating distributions $(P^{M_1},Q^{M_2})\in\calP_l^K$. It follows from the weak law of large numbers that when $n$ is sufficiently large, for each $(i,j)\in[M_1]\times[M_2]$, the type $\hatT_{X_i^{\xi_n}}$ converges to $P_i$ and the type $\hatT_{Y_j^{\chi_n}}$ converges to $Q_j$. Under hypothesis $\rmH_l^K$, for each $(i,j)\in\calM_l^K$, the $i$-th sequence from the first database is matched to the $j$-th sequence of the second database, i.e., $P_i=Q_j$. Thus, the scoring function satisfies
\begin{align}
\rmS_l^K(\bX^{\xi_n},\bY^{\chi_n})
&=\sum_{(i,j)\in\calM_l^K}\mathrm{GJS}(\hatT_{X_i^{\xi_n}},\hatT_{Y_j^{\chi_n}},\alpha,\beta)\\
&\to\sum_{(i,j)\in\calM_l^K}\mathrm{GJS}(P_i,Q_j,\alpha,\beta)\label{almostsure:conve1}\\
&=0\label{eq:},
\end{align}
where \eqref{almostsure:conve1} denotes convergence in probability and holds due to the weak law of large numbers and the continuous property of $\mathrm{GJS}(P,Q,\alpha,\beta)$ and 
\eqref{eq:} holds since $P_i=Q_j$ for $(i,j)\in\calM_l^K$. Analogously, for any $t\in[T_K]$ such that $t\neq l$, the scoring function $\rmS_t^K(\bx^{\xi_n},\by^{\chi_n})$ satisfies
\begin{align}
\rmS_t^K(\bX^{\xi_n},\bY^{\chi_n})
&=\sum_{(i,j)\in\calM_t^K}\mathrm{GJS}(\hatT_{X_i^{\xi_n}},\hatT_{Y_j^{\chi_n}},\alpha,\beta)\\
&\to\sum_{(i,j)\in\calM_t^K}\mathrm{GJS}(P_i,Q_j,\alpha,\beta)\\
&=\sum_{\substack{(i,j)\in(\calM_t^K\setminus\calM_l^K)}}\mathrm{GJS}(P_i,Q_j,\alpha,\beta)\label{almostsure:conve2}\\
&>0\label{ineq},
\end{align}
where \eqref{ineq} holds since $P_i=Q_j$ only for $(i,j)\in\calM_l^K$. 

Therefore, with a proper choice of $f(n)$ such that $\rmS_l^K(\bx^{\xi_n},\by^{\chi_n})<f(n)$ holds almost surely, our sequential test could make a reliable decision asymptotically. In the following subsection, we show that our sequential test is exponentially consistent and has bounded expected stopping time under any tuple of generating distributions.

\subsection{Result and Discussions}
Fix any $K\in[M_2]$, $l\in[T_K]$ and tuple of generating distributions $(P^{M_1},Q^{M_2})\in\calP_l^K$. Define the exponent function
\begin{align}
E_\rms(l,K,P^{M_1},Q^{M_2})
&:=\min_{t\in[T_K]:~t\neq l}\sum_{(i,j)\in\calM_t^K\setminus\calM_l^K}\alpha D_{\frac{\beta}{\alpha+\beta}}(Q_j\|P_i)\label{de:e:sequential}.
\end{align}

\begin{theorem}
\label{result:known}
Our sequential test satisfies the expected stopping time universality constraint and the mismatch exponent of our test satisfies
\begin{align}
\liminf_{N\to\infty}\frac{-\log\beta(\Phi|P^{M_1},Q^{M_2})}{N}\ge E_\rms(l,K,P^{M_1},Q^{M_2}) \label{ach:known}.
\end{align}
Conversely, for any sequential test $\tPhi$ satisfying the expected stopping time universality constraint, the mismatch exponent satisfies
\begin{align}
\limsup_{N\to\infty}\frac{-\log\beta(\tPhi|P^{M_1},Q^{M_2})}{N}\le E_\rms(l,K,P^{M_1},Q^{M_2}).\label{con:ei}
\end{align}
\end{theorem}

The proof of Theorem \ref{result:known} is provided in Section \ref{proof:known}.

The achievability proof of Theorem \ref{result:known} analyzes the expected stopping time and the achievable mismatch exponent of the sequential test in Section \ref{asymp:intuition} and thus demonstrates its asymptotic optimality. In particular, we show that our sequential test has bounded expected stopping time for any tuple of generating distributions $(P^{M_1},Q^{M_2})$. This property is highly desired since one drawback of sequential test is its potential unbounded expected stopping time for some tuple of generating distributions. In the converse part, we adapt the converse proof for binary classification~\cite{Ihwang2022sequential} and use the data processing inequality for KL divergence to upper bound the mismatch exponent for any sequential test with bounded expected stopping time.

Theorem \ref{result:known} shows that the mismatch exponent is strictly positive for any $l\in[T_K]$ under any tuple of generating distributions $(P^{M_1},Q^{M_2})\in\calP_l^K$, when $(\alpha,\beta)$ are both positive. It follows from the definition of R\'enyi divergence in \eqref{def:renyi} that 
\begin{align}
\alpha D_{\frac{\beta}{\alpha+\beta}}(Q_j\|P_i)
&=\frac{\alpha }{\frac{\beta}{\alpha+\beta}-1}\log\sum_{x\in\calX}Q_j(x)^{\frac{\beta}{\alpha+\beta}}P_i(x)^{\frac{\alpha}{\alpha+\beta}}\\
&=\frac{\beta}{\frac{\alpha}{\alpha+\beta}-1}\log\sum_{x\in\calX}P_i(x)^{\frac{\alpha}{\alpha+\beta}}Q_j(x)^{\frac{\beta}{\alpha+\beta}}\\
&=\beta D_{\frac{\alpha}{\alpha+\beta}}(P_i\|Q_j)\label{renyi:twoform}.
\end{align}
Since R\'enyi divergence is non-decreasing in its parameter, it follows that for any $\alpha$, $\alpha D_{\frac{\beta}{\alpha+\beta}}(Q_j\|P_i)$ is non-decreasing in $\beta$ since $\frac{\beta}{\alpha+\beta}=1-\frac{\alpha}{\alpha+\beta}$ is non-decreasing in $\beta$. Analogously, for any $\beta$, $\beta D_{\frac{\alpha}{\alpha+\beta}}(P_i\|Q_j)$ is non-decreasing in $\alpha$. Thus, the mismatch exponent is non-decreasing in both $\alpha$ and $\beta$. This is consistent with our intuition: with more samples, it is easier to make a reliable decision. In the extreme case when either $\alpha$ or $\beta$ tends to zero, it follows that $\alpha D_{\frac{\beta}{\alpha+\beta}}(Q_j\|P_i)=0$, leading to a zero mismatch exponent. This is because with almost no samples, it is impossible to make a reliable decision. For any finite $\alpha$, if $\beta\to\infty$, it follows that $\alpha D_{\frac{\beta}{\alpha+\beta}}(Q_j\|P_i)\to \alpha D(Q_j\|P_i)$, which implies that the mismatch exponent cannot increase without bound with respect to $\beta$ but converges to a maximum value. Similar result holds for the parameter $\alpha$ as well since $\lim_{\alpha\to\infty}\alpha D_{\frac{\beta}{\alpha+\beta}}(Q_j\|P_i)=\beta D(P_i\|Q_j)$.

When $M_2=K=1$, the above result specializes to statistical classification where one wishes to classify the generating distribution of a testing sequence by using $M_1$ training sequences. In this case, $T_K=M_1$, $\calM^K=\{(l,1)\}_{l\in[M_1]}$. The mismatch exponent satisfies that for each $l\in[M_1]$,
\begin{align}
E_\rms(l,K,P^{M_1},Q^{M_2})
&=\min_{t\in[M_1]:~t\neq l}\alpha D_{\frac{\beta}{\alpha+\beta}}(Q_1\|P_t)
=\min_{t\in[M_1]:~t\neq l}\beta D_{\frac{\alpha}{\alpha+\beta}}(P_t\|Q_1).
\end{align}
Furthermore, when $M_1=2$ and $\beta=1$, the above result specializes to the result for binary classification under the expected stopping time constraint~\cite[Theorem 1]{Ihwang2022sequential}. Thus, our results also characterize the performance of optimal sequential test that satisfies the expected stopping time constraint for statistical classification. The impact of $(\alpha,\beta)$ on the performance for this special case is exactly the same.

\subsection{Benefit of Sequentiality}
To elucidate the benefit of sequentiality, generalizing the analysis of \cite[Theorem 2]{li2014} to sequence matching, we propose a fixed-length test using the minimal scoring function decision rule, bound its achievable mismatch exponent, and show that our sequential test has strictly larger exponent than the fixed-length test.

Given any distributions $(P^{M_1},Q^{M_2})\in\calP(\calX)^{M_1+M_2}$ and $(\Omega^{M_1},\Psi^{M_2})\in\calP(\calX)^{M_1+M_2}$, define the following linear combination of KL divergences:
\begin{align}
E(P^{M_1},Q^{M_2},\Omega^{M_1},\Psi^{M_2},\alpha,\beta)
:=\sum_{i\in[M_1]}\alpha D(\Omega_i\|P_i)+\sum_{j\in[M_2]}\beta D(\Psi_j\|Q_j)\label{def:epqop}.
\end{align}
For any $l\in[T_K]$, define the exponent function
\begin{align}
E_\rmf(l,K,P^{M_1},Q^{M_2})
&:=\min_{t\in[T_K]:~t\neq l}\min_{\substack{(\Omega^{M_1},\Psi^{M_2})\in(\calP(\calX))^{M_1+M_2}:\\ \rmG_t^K(\Omega^{M_1},\Psi^{M_2},\alpha,\beta)\leq \rmG_l^K(\Omega^{M_1},\Psi^{M_2},\alpha,\beta)
}}E(P^{M_1},Q^{M_2},\Omega^{M_1},\Psi^{M_2})\label{def:exponent:fl}.
\end{align}

Fix any integer $N\in\bbN$. Consider the fixed-length test $\Phi_{\rm{FL}}=(N,\phi_N)$ such that $\phi_N(\bX^{\xi_N},\bY^{\chi_N})=\rmH_l^K$ if
\begin{align}
l=\argmin_{t\in[T_K]}\rmS_t^K(\bX^{\xi_N},\bY^{\chi_N})\label{test:fl}.
\end{align}
The asymptotic intuition why the fixed-length test works is similar to the sequential test and thus omitted. The achievable performance of the test is characterized in the following theorem.
\begin{theorem}
\label{theo:fl}
Fix any $l\in[T_K]$ and $(P^{M_1},Q^{M_2})\in\calP_l^K$. The mismatch exponent of the fixed-length test in \eqref{test:fl} satisfies
\begin{align}
\liminf_{N\to\infty}\frac{-\log \beta(\Phi|P^{M_1},Q^{M_2})}{n}\geq E_\rmf(l,K,P^{M_1},Q^{M_2}).
\end{align}
\end{theorem}
The proof of Theorem \ref{theo:fl} is analogous to that of Theorem \ref{result:known} and provided in Appendix \ref{proof:fl:test} for completeness. 

Comparing Theorems \ref{result:known} and \ref{theo:fl}, we reveal the benefit of sequentiality. Specifically, it follows from \eqref{cite4theo2} that the mismatch exponent for the sequential test satisfies
\begin{align}
E_\rms(l,K,P^{M_1},Q^{M_2})
&=\min_{t\in[T_K]:~t\neq l}\min_{\substack{(\Omega^{M_1},\Psi^{M_2})\in(\calP(\calX))^{M_1+M_2}:\\ \rmG_t^K(\Omega^{M_1},\Psi^{M_2},\alpha,\beta)\leq 0
}}E(P^{M_1},Q^{M_2},\Omega^{M_1},\Psi^{M_2},\alpha,\beta)\label{es:equivalent}.
\end{align}
It follows from \eqref{def:Gl} that $\rmG_l^K(\Omega^{M_1},\Psi^{M_2},\alpha,\beta)$ is nonnegative since it is the sum of Kl divergence terms. Thus,  
\begin{align}
\nn&\big\{(\Omega^{M_1},\Psi^{M_2})\in(\calP(\calX))^{M_1+M_2}:~\rmG_t^K(\Omega^{M_1},\Psi^{M_2},\alpha,\beta)\leq 0 \big\}\\*
&\subseteq \big\{(\Omega^{M_1},\Psi^{M_2})\in(\calP(\calX))^{M_1+M_2}:~\rmG_t^K(\Omega^{M_1},\Psi^{M_2},\alpha,\beta)\leq \rmG_l^K(\Omega^{M_1},\Psi^{M_2},\alpha,\beta)\big\}.
\end{align}
As a result, $E_\rmf(l,K,P^{M_1},Q^{M_2})\le E_\rms(l,K,P^{M_1},Q^{M_2})$. 

One might wish to compare the performance of the above test with the fixed-length test in \cite[Theorem 2]{zhou2024tit}. Denote the minimizer of \eqref{test:fl} as $i^*(\bX^{\xi_N},\bY^{\chi_N})$. Define the value of second minimal scoring function as
\begin{align}
h(\bX^{\xi_N},\bY^{\chi_N})
&:=\argmin_{t\in[T_K]:~t\neq i^*(\bX^{\xi_N},\bY^{\chi_N})}\rmS_t^K(\bX^{\xi_N},\bY^{\chi_N}).
\end{align}
Fix any $l\in[T_K]$ and $\lambda\in\bbR_+$. The test $\Phi_{\rm{Zhou}}=(N,\phi_{\rm{Zhou}})$ in \cite[Theorem 2]{zhou2024tit} operates as follows:
\begin{align}
\phi_{\rm{Zhou}}(\bX^{\xi_N},\bY^{\chi_N})
&=\left\{
\begin{array}{ll}
\rmH_l^K&\mathrm{if~}i^*(\bX^{\xi_N},\bY^{\chi_N})=l,~h(\bX^{\xi_N},\bY^{\chi_N})>\lambda\\
\rmH_\rmr&\mathrm{if~}h(\bX^{\xi_N},\bY^{\chi_N})\leq \lambda.
\end{array}
\right.
\end{align}
For any $(P^{M_1},Q^{M_2})\in\calP_l^K$, it follows from \eqref{def:mismatch} that the mismatch probability of the above test satisfies
\begin{align}
\beta(\Phi_{\rm{Zhou}}|P^{M_1},Q^{M_2})
&=\bbP_l^K\big\{\phi_{\rm{Zhou}}(\bX^{\xi_N},\bY^{\chi_N})\neq \rmH_l^K\big\}\\
&=\bbP_l^K\Big\{\big(i^*(\bX^{\xi_N},\bY^{\chi_N})\neq l,~h(\bX^{\xi_N},\bY^{\chi_N})>\lambda\big)~\mathrm{or~}\big(h(\bX^{\xi_N},\bY^{\chi_N})\leq \lambda\big)\Big\}\\
&\geq \bbP_l^K\Big\{\big(i^*(\bX^{\xi_N},\bY^{\chi_N})\neq l,~h(\bX^{\xi_N},\bY^{\chi_N})>\lambda\big)~\mathrm{or~}\big(i^*(\bX^{\xi_N},\bY^{\chi_N})\neq l,~h(\bX^{\xi_N},\bY^{\chi_N})\leq \lambda\big)\Big\}\\
&=\bbP_l^K\Big\{\big(i^*(\bX^{\xi_N},\bY^{\chi_N})\neq l\Big\}\\
&=\beta(\Phi_{\rm{FL}}|P^{M_1},Q^{M_2}).
\end{align}
Thus, the fixed-length test in \eqref{test:fl} has better performance than the fixed-length test in \cite[Theorem 2]{zhou2024tit} when the null decision is considered as a mismatch error event.

\section{Main Results for the Case of Unknown Number of Matches}
\label{sec:results:unknown}

\subsection{Sequential Test}
\label{sec:seq_test:uk}
Recall that $(\alpha,\beta)\in\bbR_+^2$ are two positive real numbers. For each $n\in\bbN$, recall that $\xi_n=\lceil \alpha n\rceil$ and $\chi_n=\lceil \beta n\rceil$. Consider any realizations of two databases $\bx^{\xi_n}=\{x_1^{\xi_n},\ldots,x_{M_1}^{\xi_n}\}$ and $\by^{\chi_n}=\{y_1^{\chi_n},\ldots,y_{M_2}^{\chi_n}\}$. Recall that $\calM^K=\{\calM_l^K\}_{l\in[T_K]}$ collect all possibilities of matched pairs when there are $K$ matches. Let $\calM:=\{\calM^K\}_{K\in[M_2]}$ collect all possibilities of matched pairs when there is at least one matched pair. Recall the definition of the scoring function $\rmS_t^K(\bx^{\xi_n},\by^{\chi_n})$ in \eqref{def:Sl} for each $t\in[T_K]$. 

Fix three positive real numbers $(\lambda_1,\lambda_2,\lambda_3)\in\bbR_+$ such that $\lambda_2\leq \min\{\lambda_1,\lambda_3\}$. For each integer $n\in\bbN$, define the event
\begin{align}
\calA^n&:=\Big\{\forall~(h,t)\in\calM,~\rmS_t^h(\bX^{\xi_n},\bY^{\chi_n})>\lambda_1\Big\}\label{def:calA1}.
\end{align}
For simplicity, given any $(K,l)\in\calM$, let $\calM_{\setminus l}^{\setminus K}:=\{(h,t)\in\calM:~(h,t)\neq (K,l)\}$. Fix any $(h,t)\in\calM$, define the following events:
\begin{align}
\calB_{1,h,t}^n&:=\big\{\rmS_t^h(\bX^{\xi_n},\bY^{\chi_n})\leq \lambda_2\big\},\label{def:calB1ht}\\
\calB_{2,h,t}^n&:=\Big\{\min_{\bart\in[T_h]:~\bart\neq t}\rmS_{\bart}^h(\bX^{\xi_n},\bY^{\chi_n})>\lambda_3\Big\},\label{def:calB2ht}\\
\calB_{h,t}^n&:=\calB_{1,h,t}^n\cap\calB_{2,h,t}^n,\label{def:calBht}\\
\calB^n&:=\bigcup_{(K,l)\in\calM}\bigg(\calB_{K,l}^n \bigcap\Big(\bigcap_{(h,t)\in\calM_{\setminus l}^{\setminus K}}(\calB_{h,t}^n)^\rmc\Big)\bigg)     \label{def:calB}.
\end{align}

Fix an integer $N\in\bbN$. The stopping time $\tau$ is defined as
\begin{align}
\tau\
&:=\inf\big\{n\in\bbN_{N-1}:~\calA^n\cup\calB^n\big\}\label{def:tau:uk}.
\end{align}
At the stopping time, the test $\phi_\tau$ operates as follows:
\begin{align}
\phi_\tau(\bX^{\xi_\tau},\bY^{\chi_\tau})
&=\left\{
\begin{array}{ll}
\rmH_l^K,~(K,l)\in\calM&\mathrm{if~}\calB_{K,l}^n\bigcap_{(h,t)\in\calM_{\setminus l}^{\setminus K}}(\calB_{h,t}^n)^\rmc,\\
\rmH_\rmr&\mathrm{otherwise}.
\end{array}
\right.
\end{align}
It follows from the definitions of $\calB^n$ in \eqref{def:calB} and the stopping time $\tau$ in \eqref{def:tau:uk} that the null hypothesis $\rmH_\rmr$ is output when the event $\calA^n$ is true. Furthermore, our sequential test decides that hypothesis $\rmH_l^K$ is true if $\calB_{K,l}$ is the only true events among all events $\{\calB_{h,t}\}_{(h,t)\in\calM}$. 

\subsection{Asymptotic Intuition}
Before presenting the theoretical results, we first explain why the above sequential test works asymptotically. Set $N$ sufficiently large so that $n$ is also sufficiently large. We first consider the null hypothesis $\rmH_\rmr$. In this case, there is no match and  $(P^{M_1},Q^{M_2})\in\calP_0^K$, i.e., $P_i\neq Q_j$ for any $(i,j)\in[M_1]\times[M_2]$. For each $(h,t)\in\calM$, it follows from the weak law of large numbers that
\begin{align}
\rmS_t^h(\bX^{\xi_n},\bY^{\chi_n})
&=\sum_{(i,j)\in\calM_t^h}\mathrm{GJS}(\hatT_{X_i^{\xi_n}},\hatT_{Y_j^{\chi_n}},\alpha,\beta)\\
&\to\sum_{(i,j)\in\calM_t^h}\mathrm{GJS}(P_i,Q_j,\alpha,\beta)\label{intuition:hr:0}.
\end{align}
As a result,
\begin{align}
\min_{(h,t)\in\calM}\rmS_t^h(\bX^{\xi_n},\bY^{\chi_n})
&\to \min_{(h,t)\in\calM}\sum_{(i,j)\in\calM_t^h}\mathrm{GJS}(P_i,Q_j,\alpha,\beta)\\
&=\min_{(i,j)\in[M_1]\times[M_2]}\mathrm{GJS}(P_i,Q_j,\alpha,\beta)\label{eqn:67}\\
&=:\rmG_0(P^{M_1},Q^{M_2},\alpha,\beta)\label{def:G0}\\
&>0\label{intuition:hr},
\end{align}
where \eqref{eqn:67} is justified in Appendix \ref{just:eqn:67}. Thus, as long as $\lambda_1<\rmG_0(P^{M_1},Q^{M_2},\alpha,\beta)$, the probability of the event $\calA^n$ tends to one asymptotically as $n\to\infty$ under the null hypothesis, which implies that the correct decision of $\rmH_\rmr$ would be output.

Next consider non-hull hypotheses. Fix any $K\in[M_2]$ and $l\in[T_K]$. Suppose that hypothesis $\rmH_l^K$ is true. In this case, $(P^{M_1},Q^{M_2})\in\calP_l^K$. It suffices to show that the probability of the event $\calB_{K,l}^n\bigcap  \left(\bigcap_{(h,t)\in\calM_{\setminus l}^{\setminus K}}(\calB_{h,t}^n)^\rmc\right)$ tends to one as $n$ increases to infinity, which implies that the correct decision $\rmH_l^K$ would be output by our test. In other words, we need to show that the probabilities of events $(\calB_{K,l}^n)^\rmc$ and $\{\calB_{h,t}^n\}_{(h,t)\in\calM_{\setminus l}^{\setminus K}}$ all asymptotically decrease to zero. 

It follows from \eqref{eq:} that the scoring function $\rmS_l^K(\bX^{\xi_n},\bY^{\chi_n})\to 0$. Thus, as $n\to\infty$, the probability of the event $(\calB_{1,K,l})^\rmc$ vanishes when $\lambda_2>0$. Furthermore, it follows from \eqref{intuition:hr:0} and $(P^{M_1},Q^{M_2})\in\calP_l^K$ that
\begin{align}
\min_{t\in[T_K]:t\neq l}\rmS_t^K(\bX^{\xi_n},\bY^{\chi_n})
&\to \min_{t\in[T_K]:t\neq l}\sum_{(i,j)\in(\calM_t^K\setminus\calM_l^K)}\mathrm{GJS}(P_i,Q_j,\alpha,\beta)\\
&=:\Lambda_l^K(P^{M_1},Q^{M_2},\alpha,\beta)\label{def:Lambdalk}\\
&>0\label{intuition:hlk3},
\end{align}
Thus, when $\lambda_3<\Lambda_l^K(P^{M_1},Q^{M_2},\alpha,\beta)$, the probability of $(\calB_{2,K,l}^n)^\rmc$ vanishes as $n\to\infty$. As a result, it follows from the definition of $\calB_{K,l}$ in \eqref{def:calBht} that when $\lambda_2>0$ and $\lambda_3<\Lambda_l^K(P^{M_1},Q^{M_2},\alpha,\beta)$, the probability of $(\calB_{K,l}^n)^\rmc$ vanishes as $n\to\infty$.

The analysis of $\{\calB_{h,t}^n\}_{(h,t)\in\calM_{\setminus l}^{\setminus K}}$ is separated into three cases.
\begin{itemize}
\item Case i) : $(h,t)\in\calM$ such that $h=K$ and $t\neq l$. It follows from \eqref{eq:} and $(P^{M_1},Q^{M_2})\in\calP_l^K$ that
\begin{align}
\min_{\bart\in[T_h]:~\bart\neq t}\rmS_{\bart}^h(\bX^{\xi_n},\bY^{\chi_n})
\leq \rmS_l^K(\bX^{\xi_n},\bY^{\chi_n})\to 0.
\end{align}
Thus, the probability of $\calB_{h,t}^n$ decreases to zero asymptotically if $\lambda_3>0$ since the event $\calB_{2,h,t}^n$ has vanishing probability.

\item Case ii) :  $(h,t)\in\calM$ such that $h>K$. It follows from \eqref{intuition:hr:0} and $(P^{M_1},Q^{M_2})\in\calP_l^K$ that
\begin{align}
\rmS_t^h(\bX^{\xi_n},\bY^{\chi_n})
&\geq 
\min_{\bart\in[T_h]}\rmS_{\bart}^h(\bX^{\xi_n},\bY^{\chi_n})\\
&\to \min_{\bart\in[T_h]}\sum_{(i,j)\in\calM_{\bart}^h:~(i,j)\notin\calM_l^K}\mathrm{GJS}(P_i,Q_j,\alpha,\beta)\\
&\geq \min_{\substack{h\in[M_2]:\\h>K}}\min_{\bart\in[T_h]}\sum_{(i,j)\in\calM_{\bart}^h:~(i,j)\notin\calM_l^K}\mathrm{GJS}(P_i,Q_j,\alpha,\beta)\\
&=\min_{\bart\in[T_{K+1}]}\sum_{(i,j)\in\calM_{\bart}^{K+1}:~(i,j)\notin\calM_l^K}\mathrm{GJS}(P_i,Q_j,\alpha,\beta)\label{usemyresult}\\
&=:\kappa_l^K(P^{M_1},Q^{M_2},\alpha,\beta)\label{def:kappalk} \\
& >0,
\end{align}
where \eqref{usemyresult} follows from the result in \cite[Eq. (63)]{zhou2024tit}. Thus, the probability of $\calB_{h,t}^n$ decreases to zero asymptotically if $\lambda_2<\kappa_l^K(P^{M_1},Q^{M_2},\alpha,\beta)$ since the event $\calB_{1,h,t}^n$ has vanishing probability.

\item Case iii) : $(h,t)\in\calM$ such that $h<K$. Note that this case occurs if $K\geq 2$. When there are $K\geq 2$ matches, for any $h<K$, one can find $(t_1,t_2)\in[T_h]^2$ such that $t_1\neq t_2$, $\calM_{t_1}^h\subset\calM_l^K$ and $\calM_{t_2}^h\subset\calM_l^K$. As a result, it follows from \eqref{intuition:hr:0} and $(P^{M_1},Q^{M_2})\in\calP_l^K$ that
\begin{align}
\rmS_{t_1}^h(\bX^{\xi_n},\bY^{\chi_n})&\to 0,\\
\rmS_{t_2}^h(\bX^{\xi_n},\bY^{\chi_n})&\to 0.
\end{align}
Since either $t_1\neq t$ or $t_2\neq t$, it follows that
\begin{align}
\min_{\bart\in[T_h]:~\bart\neq t}\rmS_{\bart}^h(\bX^{\xi_n},\bY^{\chi_n})
\leq \max\big\{\rmS_{t_1}^h(\bX^{\xi_n},\bY^{\chi_n}),\rmS_{t_2}^h(\bX^{\xi_n},\bY^{\chi_n})\big\}\to 0.
\end{align}
Thus, the probability of $\calB_{h,t}^n$ decreases to zero asymptotically if $\lambda_3>0$ since the event $\calB_{2,h,t}^n$ has vanishing probability.
\end{itemize}

Combining the above analyses, we conclude that i) under the null hypothesis $\rmH_\rmr$, our sequential test could make a correct decision asymptotically if $0<\lambda_1<\rmG_0(P^{M_1},Q^{M_2},\alpha,\beta)$ for any $(P^{M_1},Q^{M_2})\in\calP_0$, and ii) for any $(K,l)\in\calM$, under the non-null hypothesis $\rmH^K_l$, our sequential test could make a correct decision asymptotically if $0<\lambda_2<\kappa_l^K(P^{M_1},Q^{M_2},\alpha,\beta)$ and $0<\lambda_3<\Lambda_l^K(P^{M_1},Q^{M_2},\alpha,\beta)$ for any $(P^{M_1},Q^{M_2})\in\calP_l^K$.

Finally, we remark that both $\Lambda_l^K(P^{M_1},Q^{M_2},1,1)$ in \eqref{def:Lambdalk} and $\kappa_l^K(P^{M_1},Q^{M_2},1,1)$ in \eqref{def:Lambdalk} do not have simpler equations. This is illustrated via two numerical examples.
\begin{itemize}
\item Consider the case where $M_1=4$, $M_2=3$, and $K=2$. Set distributions $P^{M_1}=\mathrm{Bern}(0.1,0.12,0.3,0.6)$ and set $Q^{M_2}=\mathrm{Bern}(0.1,0.12,0.4)$. In this case, the hypothesis $\rmH_l^K$ with matching set $\calM_l^K=\{(1,1),(2,2)\}$ holds. It follows from the definition of $\Lambda_l^K(P^{M_1},Q^{M_2},1,1)$ in \eqref{def:Lambdalk} that $\Lambda_l^K(P^{M_1},Q^{M_2},\alpha,\beta)=0.002$, which is achieved by $\calM_t^K=\{(1,2),(2,1)\}$. This numerical example verifies that the minimization of $\Lambda_l^K(P^{M_1},Q^{M_2},1,1)$ in \eqref{def:Lambdalk} is not achieved by a set $\calM_{t^*}^K$ that differs from $\calM_l^K$ by one element.
\item Consider another set of distributions such that $P^{M_1}=\mathrm{Bern}(0.1,0.3,0.15,0.8)$ and set $Q^{M_2}=\mathrm{Bern}(0.1,0.3,0.4)$. In this case, the hypothesis $\rmH_l^K$ with the matching set $\calM_l^K=\{(1,1),(2,2)\}$ holds. It follows from the definition of $\kappa_l^K(P^{M_1},Q^{M_2},\alpha,\beta)$ in \eqref{def:kappalk} that $\kappa_l^K(P^{M_1},Q^{M_2},1,1)=0.0438$, which is achieved by $\calM_{\bart}^{K+1}=\{(1,1),(2,3),(3,2)\}$. This numerical example verifies that the minimization of $\kappa_l^K(P^{M_1},Q^{M_2},1,1)$ in \eqref{def:Lambdalk} is not achieved by a set $\calM_{\bart^*}^{K+1}$ that differs from $\calM_l^K$ by one element, i.e., 
\begin{align}
\kappa_l^K(P^{M_1},Q^{M_2},1,1)\neq \min_{(i,j)\in[M_1]\times[M_2]:~i\notin\calC_t^K,~j\notin\calD_t^K}\mathrm{GJS}(P_i,Q_j,\alpha,\beta)=0.0806\label{wrong:eq}.
\end{align}
Note that \eqref{wrong:eq} was incorrectly claimed to hold with equality in \cite[Eq. (64)]{zhou2024tit}.
\end{itemize}

\subsection{Results and Discussions}
Recall the definition of $E(\cdot)$ in \eqref{def:epqop}. Fix any $\lambda\in\bbR_+$. Given any tuple of distributions $(P^{M_1},Q^{M_2})\in\calP(\calX)^{M_1+M_2}$, define the following exponent function
\begin{align}
E_\rmr(\lambda,P^{M_1},Q^{M_2})
:=\min_{(h,t)\in\calM}\min_{\substack{(\Omega^{M_1},\Psi^{M_2})\in(\calP(\calX))^{M_1+M_2}:\\\rmG_t^h(\Omega^{M_1},\Psi^{M_2},\alpha,\beta)\leq \lambda
}}E(P^{M_1},Q^{M_2},\Omega^{M_1},\Psi^{M_2},\alpha,\beta)\label{def:er}.
\end{align}
As we shall show, the exponent function $E_\rmr(\lambda,P^{M_1},Q^{M_2})$ characterizes the false alarm exponent. Furthermore, when $E_\rmr(\lambda,P^{M_1},Q^{M_2})$ is strictly positive, the expected stopping time of our sequential test is bounded for any $(P^{M_1},Q^{M_2})\in\calP_0$ when $N$ is sufficiently large.

Fix any $(K,l)\in\calM$ and $(P^{M_1},Q^{M_2})\in\calP_l^K$. 
Define another exponent function
\begin{align}
F(\lambda,P^{M_1},Q^{M_2})
:=\min_{(t_1,t_2)\in[T_K]^2:~t_1\neq t_2}\min_{\substack{(\Omega^{M_1},\Psi^{M_2})\in(\calP(\calX))^{M_1+M_2}:\\ \rmG_{t_1}^K(\Omega^{M_1},\Psi^{M_2},\alpha,\beta)\leq \lambda\\
\rmG_{t_2}^K(\Omega^{M_1},\Psi^{M_2},\alpha,\beta)\leq \lambda}}
E(P^{M_1},Q^{M_2},\Omega^{M_1},\Psi^{M_2},\alpha,\beta)\label{def:exponent:nonnull}.
\end{align}
As we shall show in the proof, when $F(\lambda_3,P^{M_1},Q^{M_2})$ is strictly positive, the expected stopping time of our sequential test is bounded for any $(P^{M_1},Q^{M_2})\in\calP_l^K$ when $N$ is sufficiently large.

Finally, define the following exponent function
\begin{align}
G(\lambda,P^{M_1},Q^{M_2})
:=\min_{(h,t)\in\calM:~h>K}\min_{\substack{(\Omega^{M_1},\Psi^{M_2})\in(\calP(\calX))^{M_1+M_2}:\\\rmG_t^h(\Omega^{M_1},\Psi^{M_2},\alpha,\beta)\leq \lambda
}}E(P^{M_1},Q^{M_2},\Omega^{M_1},\Psi^{M_2},\alpha,\beta)\label{def:g:exponent}.
\end{align} 
As we show below, 
$G(\lambda_2,P^{M_1},Q^{M_2})$ is related to the mismatch exponent.

The properties of the above exponent functions are summarized in the following lemma.
\begin{lemma}
\label{prop:exponents}
The following claims hold.
\begin{enumerate}
\item For any $(P^{M_1},Q^{M_2})\in\calP_0$, the exponent function $E_\rmr(\lambda,P^{M_1},Q^{M_2})$ is non-increasing in $\lambda$ and equals zero when $\lambda\geq \rmG_0(P^{M_1},Q^{M_2},\alpha,\beta)$ (cf. \eqref{def:G0}).
\item For any $(P^{M_1},Q^{M_2})\in\calP_l^K$, the exponent function $F(\lambda,P^{M_1},Q^{M_2})$ is non-increasing in $\lambda$ and equals zero if $\lambda\geq \Lambda_l^K(P^{M_1},Q^{M_2},\alpha,\beta)$ (cf. \eqref{def:Lambdalk}).
\item For any $(P^{M_1},Q^{M_2})\in\calP_l^K$, the exponent function $G(\lambda,P^{M_1},Q^{M_2})$ is non-increasing in $\lambda$ and equals zero if 
$\lambda\geq \kappa_l^K(P^{M_1},Q^{M_2},\alpha,\beta)$ (cf. \eqref{def:kappalk}). 
\end{enumerate}
\end{lemma}
\begin{proof}
Claim (ii) was proved in \cite[Claim (ii) of Lemma 1]{zhou2024tit} and Claim (iii) was proved in \cite[Claim (ii) of Lemma 2]{zhou2024tit}. It suffices to prove Claim (i).

It follows from \eqref{def:er} that $E_\rmr(\lambda,P^{M_1},Q^{M_2})$ is non-increasing in $\lambda$. Furthermore, $E_\rmr(\lambda,P^{M_1},Q^{M_2})=0$ if there exists $(h,t)\in\calM$ such that $\rmG_t^h(P^{M_1},Q^{M_2},\alpha,\beta)\leq \lambda$. This corresponds to
\begin{align}
\lambda
&\geq \min_{(h,t)\in\calM}\rmG_t^h(P^{M_1},Q^{M_2},\alpha,\beta)\\
&=\rmG_0(P^{M_1},Q^{M_2},\alpha,\beta)\label{use:G0:prop}
\end{align}
where \eqref{use:G0:prop} follows from the result in \eqref{def:G0}.
\end{proof}

Recall that $(\lambda_1,\lambda_2,\lambda_3)\in\bbR_+^3$ and $\lambda_2\leq \min\{\lambda_1,\lambda_3\}$.
\begin{theorem}
\label{result:unknown}
Our sequential test ensures that
\begin{enumerate}
\item For any $(P^{M_1},Q^{M_2})\in\calP_0$, the expected stopping time $\bbE[\tau]\leq N$ when $N$ is sufficiently large if $\lambda_1<\rmG_0(P^{M_1},Q^{M_2},\alpha,\beta)$ and the false alarm exponent satisfies \begin{align}
\liminf_{N\to\infty}\frac{-\log \eta(\Phi|P^{M_1},Q^{M_2})}{N}\geq E_\rmr(\lambda_1,P^{M_1},Q^{M_2}).
\end{align}
\item For any $(K,l)\in\calM$ and any $(P^{M_1},Q^{M_2})\in\calP_l^K$, the expected stopping time $\bbE[\tau]\leq N$ when $N$ is sufficiently large if $\lambda_2>0$  and $\lambda_3<\Lambda_l^K(P^{M_1},Q^{M_2},\alpha,\beta)$ and
\begin{itemize}
\item the mismatch exponent satisfies
\begin{align}
\liminf_{N\to\infty}\frac{-\log\bar{\beta}(\Phi|P^{M_1},Q^{M_2})}{N}\geq \min\Big\{G(\lambda_2,P^{M_1},Q^{M_2}),\lambda_3\Big\}.
\end{align}
\item the false reject exponent satisfies
\begin{align}
\liminf_{N\to\infty}\frac{-\log\zeta(\Phi|P^{M_1},Q^{M_2})}{N}\geq \lambda_1.
\end{align}
\end{itemize}
\end{enumerate}
\end{theorem}
The proof of Theorem \ref{result:unknown} is provided in Section \ref{proof:unknown}. In particular, we upper bound the expected stopping time and lower bound the exponents of all three error probabilities. The parameter $\lambda_1$ is critical for the null hypothesis, which should be smaller than $\rmG_0(P^{M_1},Q^{M_2},\alpha,\beta)$ to ensure that the expected stopping time under the null hypothesis is bounded and the false alarm exponent $E_\rmr(\lambda_1,P^{M_1},Q^{M_2})$ is positive when the generating distributions are $(P^{M_1},Q^{M_2})$. Since the generating distributions are unknown, choosing a parameter $\lambda_1$ guarantees the performance under the null hypothesis for the set of distributions $(P^{M_1},Q^{M_2})$ such that $\rmG_0(P^{M_1},Q^{M_2},\alpha,\beta)>\lambda_1$. A smaller $\lambda_1$ leads to guaranteed bounded expected stopping time for a larger set of generating distributions and a larger false alarm exponent. However, $\lambda_1$ cannot be too small since it directly lower bounds the false reject exponent.

The parameters $(\lambda_2,\lambda_3)$ are critical for bounding the expected stopping time and the mismatch exponent under each non-null hypothesis. In particular, the expected stopping time is bounded if $\lambda_2$ is strictly positive and $\lambda_3$ is not larger than $\Lambda_l^K(P^{M_1},Q^{M_2},\alpha,\beta)$. The mismatch exponent consists of two parts: $G(\lambda_2,P^{M_1},Q^{M_2})$ characterizes the exponential decay rate for the probability of overestimating the number of matches and $\lambda_3$ characterizes the exponential decay rate of identifying a wrong set of matches. The mismatch exponent increases in $\lambda_3$ and decreases in $\lambda_2$. Therefore, a good performance is guaranteed when $\lambda_2$ is small but $\lambda_3$ is large under a given tuple of distributions. Asymptotically, one can choose $\lambda_2$ to be arbitrarily close to zero to achieve the largest mismatch exponent. However, $\lambda_3$ cannot be too large since the expected stopping time is bounded only if $\lambda_3<\Lambda_l^K(P^{M_1},Q^{M_2},\alpha,\beta)$. Thus, there is a tradeoff between guaranteeing a bounded expected stopping time and achieving largest mismatch exponent via the choice of $\lambda_3$.

We can now specialize the above results to statistical classification when $M_2=K=1$.  The task is to decide whether the sequence $Y_1^{\chi_n}$ is matched to any of the sequences $\bX^{\xi_n}=(X_1^{\xi_n},\ldots,X_{M_1}^{\xi_n})$. In this case, $T_K=M_1$, $\calM=\{(t,1)\}_{t\in[M_1]}$ and there is at most one match. It follows that the test in Section \ref{sec:seq_test:uk} satisfies Theorem \ref{result:unknown} except that the mismatch exponent is replaced by $\lambda_3$. This is because the error event of overestimating the number of matches no longer occurs, which leads to the disappearance of the exponent function $G(\lambda_2,P^{M_1},Q_1)$. Furthermore, the expressions of the exponents are simplified significantly. For any generating distributions $(P^{M_1},Q_1)\in\calP_0$, the false alarm exponent is given by
\begin{align}
E_\rmr(\lambda_1,P^{M_1},Q_1)
:=\min_{t\in[M_2]}\min_{\substack{(\Omega^{M_1},\Psi)\in(\calP(\calX))^{M_1+1}:\\\mathrm{GJS}(\Omega_i,\Psi,\alpha,\beta)\leq \lambda_1
}}\Big(\sum_{i\in[M_1]}\alpha D(\Omega\|P_i)+\beta D(\Psi\|Q_1)\Big)
\end{align}
Fix any $l\in[M_1]$. For any $(P^{M_1},Q^{M_2})\in\calP_l^K$, if
\begin{align}
\lambda_3
&<\Lambda_l^K(P^{M_1},Q_1,\alpha,\beta)=\min_{t\in[M_1]:t\neq l}\mathrm{GJS}(P_t,Q_1,\alpha,\beta),
\end{align}
the mismatch exponent is lower bounded by $\lambda_3$.

\subsection{Benefit of Sequentiality}
In \cite[Theorem 4]{zhou2024tit}, the authors proposed a two-step fixed-length test and characterized the achievable error exponents of the test. Specifically, the test first estimates the number of matches and subsequently identifies the number of matches if the estimated number of matches is positive. The results in \cite[Theorem 4]{zhou2024tit} was simplified from the original equation. For ease of comparison with our results, we use the original form of exponents in the proof of \cite[Theorem 4]{zhou2024tit}.

Let $\Phi_{\rm{Zhou}}^{\rm{uk}}$ denote the fixed-length test for \cite[Theorem 4]{zhou2024tit} when the number of matches is unknown. Recall the definitions of the exponent functions $E_\rmr(\cdot)$ in \eqref{def:er}, $F(\cdot)$ in \eqref{def:exponent:nonnull} and $G(\cdot)$ in \eqref{def:g:exponent}. It was shown in~\cite[Theorem 4]{zhou2024tit} that the fixed-length test achieves the following performance.
\begin{theorem}
\label{ld:unknown}
Given any positive real numbers $(\lambda_1',\lambda_2')\in\bbR_+^2$, there exists a fixed-length test such that
\begin{enumerate}
\item for any tuple of distributions $(P^{M_1},Q^{M_2})\in\calP_0$, the false alarm exponent satisfies
\begin{align}
\liminf_{n\to\infty}-\frac{1}{n}\log\eta(\Phi_{\rm{Zhou}}^{\rm{uk}}|P^{M_1},Q^{M_2})&\geq E_\rmr(\lambda_1',P^{M_1},Q^{M_2}).
\end{align}
\item for any tuple of distributions $(P^{M_1},Q^{M_2})\in\calP_l^K$,
\begin{itemize}
\item the mismatch exponent satisfies
\begin{align}
\liminf_{n\to\infty}-\frac{1}{n}\log\beta(\Phi_{\rm{Zhou}}^{\rm{uk}}|P^{M_1},Q^{M_2})
&\geq  \min\big\{\lambda_1',\lambda_2',G(\lambda_1',P^{M_1},Q^{M_2})\big\}.
\end{align}
\item the false reject exponent satisfies
\begin{align}
\liminf_{n\to\infty}-\frac{1}{n}\log\zeta(\Phi_{\rm{Zhou}}^{\rm{uk}}|P^{M_1},Q^{M_2})
&\geq \min\big\{\lambda_1',G(\lambda_1',P^{M_1},Q^{M_2}),F(\lambda_2',P^{M_1},Q^{M_2})\big\}.
\end{align}
\end{itemize}
\end{enumerate}
\end{theorem}
\begin{proof}
Fix any $(K,l)\in\calM$ and $(P^{M_1},Q^{M_2})\in\calP_l^K$. The mismatch exponent follows from \cite[Eq. (95), (201), (204), (208), (210)]{zhou2024tit}, the false reject exponent follows from \cite[Eq. (105), (110), (225)]{zhou2024tit} and  the false alarm exponent follows from \cite[Eq. (229), (232), (233)]{zhou2024tit}.
\end{proof}

We now discuss the benefit of sequentiality. Set $\lambda_1=\lambda_1'$, $\lambda_3=\lambda_2'$ and consider any $\lambda_2<\lambda_1'$. It follows from Theorems \ref{result:unknown} and \ref{ld:unknown} that both our sequential test and the fixed-length test $\Phi_{\rm{Zhou}}$ achieve the same false alarm exponent while our sequential test has larger Bayesian exponent under each non-null hypothesis. Note that the Bayesian exponent under each non-null hypothesis equals the minimum of the mismatch and false reject exponents. Specifically, Theorem \ref{result:unknown} implies that the Bayesian exponent of our sequential test satisfies
\begin{align}
\min\big\{G(\lambda_2,P^{M_1},Q^{M_2}),\lambda_1,\lambda_3\big\}
&\geq \min\Big\{G(\lambda_1',P^{M_1},Q^{M_2}),\lambda_1',\lambda_2'\Big\}\label{use:g:nonin}\\
&\geq  \min\Big\{\lambda_1',\lambda_2',G(\lambda_1',P^{M_1},Q^{M_2}),F(\lambda_2',P^{M_1},Q^{M_2})\Big\},
\end{align}
where \eqref{use:g:nonin} follows since the function $G(\cdot)$ is non-increasing in $\lambda$ and $\lambda_2<\lambda_1'$.

\subsection{One-Step Fixed-Length Test}
The performance of \cite[Theorem 4]{zhou2024tit} is achieved by a two-phase test~\cite[Algorithm 2]{zhou2024tit} that first estimates the number of matches and subsequently identifies the matched pairs of sequences if the estimated number of matches is positive. One might wonder whether there exists a one-step test that can achieve the same or even better performance. In the following, we answer this question affirmatively.

Recall the definition of $\calB_{h,t}^n$ in \eqref{def:calBht}. Fix any $N\in\bbN$. Inspired by the design of our sequential test, we propose the following fixed-length test $\Phi_{\rm{FL}}^{\rm{uk}}=(N,\phi_N^{\rm{uk}})$ such that for any $(K,l)\in\calM$,
\begin{align}
\phi_N^{\rm{uk}}(\bX^{\xi_N},\bY^{\chi_N})
&=
\left\{
\begin{array}{ll}
\rmH_l^K&\mathrm{if~}\calB_{K,l}^N\bigcap_{(h,t)\in\calM_{\setminus l}^{\setminus K}}(\calB_{h,t}^n)^\rmc,\\
\rmH_\rmr&\mathrm{otherwise}.
\end{array}
\right.
\label{fl:one-step}
\end{align}

\begin{theorem}
\label{theo:fl:one-step}
Given any positive real numbers $(\lambda_1,\lambda_2)\in\bbR_+^2$, the fixed-length test in \eqref{fl:one-step} ensures that
\begin{enumerate}
\item for any tuple of distributions $(P^{M_1},Q^{M_2})\in\calP_0$,
\begin{align}
\liminf_{n\to\infty}-\frac{1}{n}\log\eta(\Phi_{\rm{Zhou}}^{\rm{uk}}|P^{M_1},Q^{M_2})&\geq E_\rmr(\lambda_1,P^{M_1},Q^{M_2}).
\end{align}
\item for any tuple of distributions $(P^{M_1},Q^{M_2})\in\calP_l^K$,
\begin{itemize}
\item the mismatch exponent satisfies
\begin{align}
\liminf_{n\to\infty}-\frac{1}{n}\log\beta(\Phi_{\rm{Zhou}}^{\rm{uk}}|P^{M_1},Q^{M_2})
&\geq \min\big\{G(\lambda_1,P^{M_1},Q^{M_2}),\lambda_2\big\}.
\end{align}
\item the false reject exponent satisfies
\begin{align}
\liminf_{n\to\infty}-\frac{1}{n}\log\zeta(\Phi_{\rm{Zhou}}^{\rm{uk}}|P^{M_1},Q^{M_2})
&\geq \min\big\{\lambda_1,\lambda_2,G(\lambda_1,P^{M_1},Q^{M_2}),F(\lambda_2,P^{M_1},Q^{M_2})\big\}.
\end{align}
\end{itemize}
\end{enumerate}
\end{theorem}
The proof of Theorem \ref{theo:fl:one-step} is provided in Appendix \ref{proof:fl:one-step} for completeness. Comparing Theorems \ref{result:unknown} and \ref{theo:fl:one-step}, we reveal the benefit of sequentiality since our sequential test in Section \ref{sec:seq_test:uk} achieves the same false alarm and mismatch exponents and a larger false reject exponent than the fixed-length test in \eqref{fl:one-step}. Comparing Theorems \ref{ld:unknown} and \ref{theo:fl:one-step}, we conclude that the fixed-length test in \eqref{fl:one-step} achieves the same false alarm exponent under the null hypothesis, the same Bayesian exponent under each non-null hypothesis and a larger mismatch exponent than the two-step fixed-length test in \cite[Algorithm 2]{zhou2024tit}. In particular, the one-step fixed-length test in \eqref{fl:one-step} is simpler than the two-step fixed-length test in \cite[Algorithm 2]{zhou2024tit}.

\section{Proof for the Case of Known Number of Matches (Theorem \ref{result:known})}
\label{proof:known}

\subsection{Achievability}
Fix any $l\in[T_K]$, tuple of generating distributions $(P^{M_1},Q^{M_2})\in\calP_l^K$ and integer $N\in\bbN$. Recall that $\bbP_l^K$ denotes the joint distribution of all sequences of two databases under hypothesis $\rmH_l^K$, where $P_i=Q_j$ for any $(i,j)\in\calM_l^K$.

\subsubsection{Expected Stopping Time}
We first show that our sequential test has bounded expected stopping time. Recall the definition of $\tau$ in \eqref{def:tau}. Under hypothesis $\rmH_l^K$, it follows that
\begin{align}
\bbE_{\bbP_l^K}[\tau]
&=\sum_{n\in\bbN}\bbP_l^K\{\tau>n\}\\
&=N-1+\sum_{n\in\bbN_{N-1}}\bbP_l^K\{\tau>n\}\label{etau:step1}.
\end{align}
The second term in \eqref{etau:step1} can be further upper bounded as follows using the method of types~\cite{csiszar1998mt}:
\begin{align}
\bbP_l^K\{\tau>n\}
&=\bbP_l^K\big\{\forall~t\in[T_K],~\rmS_t^K(\bX^{\xi_n},\bY^{\chi_n})>f(n)\big\}\\
&\leq \bbP_l^K\big\{\rmS_l^K(\bX^{\xi_n},\bY^{\chi_n})>f(n)\big\}\label{etau:step1.0}\\
&=\bbP_l^K\big\{\rmG_l^K(\hatT_{\bX^{\xi_n}},\hatT_{\bY^{\chi_n}},\alpha,\beta)>f(n)\big\}\\
&=\sum_{(\bx^{\xi_n},\by^{\chi_n}):~\rmG_l^K(\hatT_{\bx^{\xi_n}},\hatT_{\by^{\chi_n}},\alpha,\beta)>f(n)}\bigg(\prod_{i\in[M_1]}P_i^{\xi_n}(x_i^{\xi_n})\bigg)\bigg(\prod_{j\in[M_2]}Q_j^{\chi_n}(y_j^{\chi_n})\bigg)
\\
&=\sum_{\substack{(\{x_i^{\xi_n},y_j^{\chi_n})\}_{(i,j)\in\calM_l^K}:\\
\sum_{(i,j)\in\calM_l^K}\rmG_l^K(\hatT_{x_i^{\xi_n}},\hatT_{y_j^{\chi_n}},\alpha,\beta)\geq f(n)}}
\prod_{(i,j)\in\calM_l^K}P_i^{\xi_n}(x_i^{\xi_n})P_i^{\chi_n}(y_j^{\chi_n})\\
&=\sum_{\substack{(\Omega^K,\Psi^K)\in(\calP^{\xi_n}(\calX))^K\times(\calP^{\chi_n}(\calX))^K:\\
\sum_{i\in[K]}\mathrm{GJS}(\Omega_i,\Psi_i,\alpha,\beta)>f(n)
}}\prod_{i\in[K]}P_i^{\xi_n}(\calT_{\Omega_i}^{\xi_n})P_i^{\chi_n}(\calT_{\Omega_j}^{\chi_n})\\
&\leq \sum_{\substack{(\Omega^K,\Psi^K)\in(\calP^{\xi_n}(\calX))^K\times(\calP^{\chi_n}(\calX))^K:\\
\sum_{i\in[K]}\mathrm{GJS}(\Omega_i,\Psi_i,\alpha,\beta)>f(n)
}}\exp\bigg(-\sum_{i\in[K]}\Big(\xi_nD(\Omega_i\|P_i)+\chi_nD(\Psi_i\|P_i)\Big)\bigg)\label{use:prob:typeclass}\\
&\leq \sum_{\substack{(\Omega^K,\Psi^K)\in(\calP^{\xi_n}(\calX))^K\times(\calP^{\chi_n}(\calX))^K:\\
\sum_{i\in[K]}\mathrm{GJS}(\Omega_i,\Psi_i,\alpha,\beta)>f(n)
}}\exp\bigg(-n\sum_{i\in[K]}\Big(\alpha D(\Omega_i\|P_i)+\beta D(\Psi_i\|P_i)\Big)\bigg)\label{use:def:xin}\\
&=\sum_{\substack{(\Omega^K,\Psi^K)\in(\calP^{\xi_n}(\calX))^K\times(\calP^{\chi_n}(\calX))^K:\\
\sum_{i\in[K]}\mathrm{GJS}(\Omega_i,\Psi_i,\alpha,\beta)>f(n)
}}\exp\bigg(-n\sum_{i\in[K]}\Big(\mathrm{GJS}(\Omega_i,\Psi_i,\alpha,\beta)+(\alpha+\beta)D(R_{\alpha,\beta}^{\Omega_i,\Psi_i}\|P_i)\Big)\label{eq:gjsab}\\
&\leq \sum_{\substack{(\Omega^K,\Psi^K)\in(\calP^{\xi_n}(\calX))^K\times(\calP^{\chi_n}(\calX))^K:\\
\sum_{i\in[K]}\mathrm{GJS}(\Omega_i,\Psi_i,\alpha,\beta)>f(n)
}}\exp\bigg(-nf(n)-n(\alpha+\beta)\sum_{i\in[K]}D(R_{\alpha,\beta}^{\Omega_i,\Psi_i}\|P_i)\bigg)\label{use:con}\\
&\leq \sum_{\substack{(\Omega^K,\Psi^K)\in(\calP^{\xi_n}(\calX))^K\times(\calP^{\chi_n}(\calX))^K}}\exp(-nf(n))\label{use:KL:nonnega}\\
&\leq (\xi_n+1)^{K|\calX|}(\chi_n+1)^{K|\calX|}\exp(-nf(n))\label{use:num:types}\\
&\leq (n\alpha+2)^{K|\calX|}(n\beta+2)^{K|\calX|}\exp(-nf(n))\label{etau:step1.1}\\
&\leq (n\alpha+2)^{-|\calX|}\label{usefn},
\end{align}
where \eqref{use:prob:typeclass} follows from the upper bound on the probability of a type class~\cite[Theorem 11.1.4]{cover2012elements}, \eqref{use:def:xin} follows since $\xi_n\leq n\alpha$ and $\chi_n\leq n\beta$, \eqref{use:con} follows since $\sum_{i\in[K]}\mathrm{GJS}(\Omega_i,\Psi_i,\alpha,\beta)>f(n)$, \eqref{use:KL:nonnega} follows since each KL divergence term $D(R_{\alpha,\beta}^{\Omega_i,\Psi_i}\|P_i)$ is non-negative, \eqref{use:num:types} follows from the upper bound on the number of types~\cite[Theorem 11.1.1]{cover2012elements}, \eqref{etau:step1.1} follows since $xi_n\leq n\alpha+1$ and $\chi_n\leq n\beta+1$, \eqref{usefn} follows from the definition of $f(n)$ in \eqref{def:fn} while the reasoning for \eqref{eq:gjsab} is as follows:
\begin{align}
\nn&\alpha D(\Omega_i\|P_i)+\beta D(\Psi_i\|P_i)\\
&=\alpha\bbE_{\Omega_i}\bigg[\log\frac{\Omega_i(X)}{P_i(X)}\bigg]+\beta \bbE_{\Psi_i}\bigg[\log\frac{\Psi_i(X)}{P_i(X)}\bigg]\\
&=\alpha\bbE_{\Omega_i}\bigg[\log\frac{\Omega_i(X)R_{\alpha,\beta}^{\Omega_i,\Psi_i}(X)}{P_i(X)R_{\alpha,\beta}^{\Omega_i,\Psi_i}(X)}\bigg]+\beta \bbE_{\Psi_i}\bigg[\log\frac{\Psi_i(X)R_{\alpha,\beta}^{\Omega_i,\Psi_i}(X)}{P_i(X)R_{\alpha,\beta}^{\Omega_i,\Psi_i}(X)}\bigg]\label{use:r1}\\
&=\alpha\bbE_{\Omega_i}\bigg[\log\frac{\Omega_i(X)}{R_{\alpha,\beta}^{\Omega_i,\Psi_i}(X)}\bigg]
+\alpha\bbE_{\Omega_i}\bigg[\log\frac{R_{\alpha,\beta}^{\Omega_i,\Psi_i}(X)}{P_i(X)}\bigg]+\beta \bbE_{\Psi_i}\bigg[\log\frac{\Psi_i(X)}{R_{\alpha,\beta}^{\Omega_i,\Psi_i}(X)}\bigg]+\beta \bbE_{\Psi_i}\bigg[\log\frac{R_{\alpha,\beta}^{\Omega_i,\Psi_i}(X)}{P_i(X)}\bigg]\label{use:r2}\\
&=\alpha D(\Omega_i\|R_{\alpha,\beta}^{\Omega_i,\Psi_i})+\beta D(\Psi_i\|R_{\alpha,\beta}^{\Omega_i,\Psi_i})+(\alpha+\beta)D(R_{\alpha,\beta}^{\Omega_i,\Psi_i}\|P_i)\\
&=\mathrm{GJS}(\Omega_i,\Psi_i,\alpha,\beta)+(\alpha+\beta)D(R_{\alpha,\beta}^{\Omega_i,\Psi_i}\|P_i),
\end{align}
where \eqref{use:r1} and \eqref{use:r2} follow from the definition of $R_{\alpha,\beta}^{\cdot}$ in \eqref{def:Rab} that specifies a distribution.

Combining \eqref{etau:step1} and \eqref{usefn}, it follows that
\begin{align}
\bbE[\tau]
&\leq N-1+\sum_{n\in\bbN_{N-1}}(n\alpha+2)^{-|\calX|}\\
&\leq N-1+\int_{(N-1)\alpha}^{\infty}(u+2)^{-|\calX|}\rmd u\\
&= N-1+\frac{(u+2)^{-|\calX|+1}}{-|\calX|+1}\bigg|_{u=(N-1)\alpha}^{\infty}\\
&=N-1+\frac{\big((N-1)\alpha+2\big)^{-|\calX|+1}}{|\calX|-1}\\
&\leq N\label{algebra},
\end{align}
where \eqref{algebra} follows since $(N-1)\alpha+2\geq 2$ and $|\calX|-1\geq 1$ when $N\geq 2$ and $|\calX|\geq 2$.

\subsubsection{Mismatch Probability}
Recall the decision rule $\phi_\tau$ in \eqref{test:tau}. Under hypothesis $\rmH_l^K$, the mismatch probability satisfies
\begin{align}
\beta(\Phi|P^{M_1},Q^{M_2})
&=\bbP_l^K\{\phi_\tau(\bX^{\xi_\tau},\bY^{\chi_\tau})\neq \rmH_l^K\}\\
&=\bbP_l^K\big\{\exists~t\in[T_K]:~t\neq l\mathrm{~and~}t\in[T_K],~\rmS_t^K(\bX^{\xi_\tau},\bY^{\chi_\tau})\leq f(\tau)\big\}\label{explainerror}\\
&\leq (T_K-1)\max_{t\in[T_K]:~t\neq l}\bbP_l^K\big\{\rmS_t^K(\bX^{\xi_\tau},\bY^{\chi_\tau})\leq f(\tau)\big\}\\
&=(T_K-1)\max_{t\in[T_K]:~t\neq l}\sum_{n\in\bbN_{N-1}}
\bbP_l^K\big\{\tau=n,~\rmS_t^K(\bX^{\xi_n},\bY^{\chi_n})\leq f(n)\big\}\\
&\leq (T_K-1)\max_{t\in[T_K]:~t\neq l}\sum_{n\in\bbN_{N-1}}
\bbP_l^K\big\{\rmS_t^K(\bX^{\xi_n},\bY^{\chi_n})\leq f(n)\big\}\label{mismatch:step1},
\end{align}
where \eqref{explainerror} holds since a mismatch error event occurs if and only if some competitive scoring function $\rmS_t^K(\bX^{\xi_n},\bY^{\chi_n})$ with $t\neq l$ is small enough with respect to $f(n)$ so that the minimal scoring function decision rule fails.

Recall the definition of $E(\cdot)$ in \eqref{def:epqop}. Define the following exponent function
\begin{align}
\Delta(n,P^{M_1},Q^{M_2})
:=\min_{t\in[T_K]:~t\neq l}\min_{\substack{(\Omega^{M_1},\Psi^{M_2})\in(\calP^{\xi_n}(\calX))^{M_1}\times(\calP^{\chi_n}(\calX))^{M_2}:\\ \rmG_t^K(\Omega^{M_1},\Psi^{M_2},\alpha,\beta)\leq f(n)
}}E(P^{M_1},Q^{M_2},\Omega^{M_1},\Psi^{M_2},\alpha,\beta)\label{def:Deltatn}.
\end{align}
Similarly to \eqref{use:prob:typeclass}, each probability term in \eqref{mismatch:step1} can be further upper bounded using the method of types as follows:
\begin{align}
\nn&\bbP_l^K\big\{\rmS_t^K(\bX^{\xi_n},\bY^{\chi_n})\leq f(n)\big\}\\
&=\sum_{\substack{(\Omega^{M_1},\Psi^{M_2})\in(\calP^{\xi_n}(\calX))^{M_1}\times(\calP^{\chi_n}(\calX))^{M_2}:\\ \rmG_t^K(\Omega^{M_1},\Psi^{M_2},\alpha,\beta)\leq f(n)
}}\Big(\prod_{i\in[M_1]}P_i^{\xi_n}(\calT_{\Omega_i}^{\xi_n})\Big)\Big(\prod_{j\in[M_2]}Q_j^{\chi_n}(\calT_{\Psi_j}^{\chi_n})\Big)\\
&\leq \sum_{\substack{(\Omega^{M_1},\Psi^{M_2})\in(\calP^{\xi_n}(\calX))^{M_1}\times(\calP^{\chi_n}(\calX))^{M_2}:\\ \rmG_t^K(\Omega^{M_1},\Psi^{M_2},\alpha,\beta)\leq f(n)
}}\exp\bigg(-\Big(\sum_{i\in[M_1]}\xi_nD(\Omega_i\|P_i)+\sum_{j\in[M_2]}\chi_nD(\Psi_j\|Q_j)\Big)\bigg)\\
&\leq \sum_{\substack{(\Omega^{M_1},\Psi^{M_2})\in(\calP^{\xi_n}(\calX))^{M_1}\times(\calP^{\chi_n}(\calX))^{M_2}}}\exp(-n\Delta(n,P^{M_1},Q^{M_2}))\\
&\leq (\xi_n+1)^{M_1}(\chi_n+1)^{M_2}\exp(-n\Delta(n,P^{M_1},Q^{M_2}))\label{useDelta},
\end{align}
where \eqref{useDelta} follows from the definition of $\Delta(\cdot)$ in \eqref{def:Deltatn}.

It follows from the definition of $f(\cdot)$ in \eqref{def:fn} that $\lim_{n\to\infty}f(n)=0$. Recall the definitions of $\calC_t^\cdot$ and $\calD_t^\cdot$ in \eqref{def:calc} and \eqref{def:cald}, respectively. It follows that
\begin{align}
\lim_{n\to\infty}\Delta(n,P^{M_1},Q^{M_2})
&=\min_{t\in[T_K]:~t\neq l}\min_{\substack{(\Omega^{M_1},\Psi^{M_2})\in(\calP(\calX))^{M_1+M_2}:\\ \forall~(i,j)\in\calM_t^K,~\Omega_i=\Psi_j}}\Big(\sum_{i\in[M_1]}\alpha D(\Omega_i\|P_i)+\sum_{j\in[M_2]}\beta D(\Psi_j\|Q_j)\Big)\label{cite4theo2}\\
&=\min_{t\in[T_K]:~t\neq l}\quad \min_{\bar{\Omega}^K\in(\calP(\calX))^K}\sum_{(i,j)\in\calM_t^K} \big(\alpha D(\bar{\Omega}_i\|P_i)+\beta D(\bar{\Omega}_i\|Q_j)\big)\label{explain1}\\
&=\min_{t\in[T_K]:~t\neq l}\sum_{(i,j)\in\calM_t^K\setminus\calM_l^K} \min_{\bar{\Omega}\in\calP(\calX)}\big(\alpha D(\bar{\Omega}\|P_i)+\beta D(\bar{\Omega}\|Q_j)\big)\label{explain2}\\
&=\min_{t\in[T_K]:~t\neq l}\sum_{(i,j)\in\calM_t^K\setminus\calM_l^K}\alpha D_{\frac{\beta}{\alpha+\beta}}(Q_j\|P_i)\label{explain4},
\end{align}
where \eqref{explain1} holds since
\begin{align}
\nn&\Big(\sum_{i\in[M_1]}\alpha D(\Omega_i\|P_i)+\sum_{j\in[M_2]}\beta D(\Psi_j\|Q_j)\Big)\\*
&=\sum_{(i,j)\in\calM_t^K}\Big(\alpha D(\Omega_i\|P_i)+\beta D(\Omega_i\|Q_j)\Big)+\sum_{i\notin\calC_t^K}\alpha D(\Omega_i\|P_i)+\sum_{j\notin\calD_t^K}
\beta D(\Psi_j\|Q_j),
\end{align}
which implies that choosing $\Omega_i=P_i$ for $i\notin\calC_t^K$ and $\Psi_j=Q_j$ for $j\notin\calD_t^K$ makes the last two terms zero since only $(\Omega_i,\Psi_j)$ with $(i,j)\in\calM_t^K$ are constrained, \eqref{explain2} follows since $P_i=Q_j$ for any $(i,j)\in\calM_l^K$, and \eqref{explain4} follows from the variational form the R\'enyi divergence~\cite[Eq. (7)]{Ihwang2022sequential} (cf. \eqref{renyi:variational}).

Combining \eqref{mismatch:step1}, \eqref{useDelta} and \eqref{explain2} leads to
\begin{align}
\liminf_{N\to\infty}\frac{-\log\beta(\Phi|P^{M_1},Q^{M_2})}{N}
&\geq \min_{(i,j)\notin\calM_l^K}\alpha D_{\frac{\beta}{\alpha+\beta}}(Q_j\|P_i).
\end{align}
The achievability proof of Theorem \ref{result:known} is now completed.

\subsection{Converse}
Given any positive real numbers $(p,q)\in(0,1)^2$, the binary KL-divergence is defined as
\begin{align}\label{binaryKL}
d(p,q):=p\log\frac{p}{q}+(1-p)\log\frac{1-p}{1-q}.
\end{align}
The first derivative of $d(p,q)$ with respect to $q$ satisfies
\begin{align}
\frac{\partial d(p,q)}{\partial q}=\frac{q-p}{q(1-q)}.
\label{fd:bkl}
\end{align}
Since $q(1-q)>0$ for any $q\in(0,1)$, $d(p,q)$ increases in $q$ when $q>p$ and decreases in $q$ when $p>q$.

Fix any $(l,t)\in[T_K]^2$ such that $t\neq l$. Let $\bbP_l$ be the joint distribution of all sequences of two databases under hypothesis $\rmH_l^K$ when $(P^{M_1},Q^{M_2})\in\calP_l^K$ and let $\tilbbP_t$ be the joint distribution of all sequences of two databases under hypothesis $\rmH_t^K$ when $(\tilP^{M_1},\tilQ^{M_2})\in\calP_t^K$. Furthermore, define the event 
\begin{align}
\calW:=\big\{\phi_\tau(\bX^{\xi_\tau},\bY^{\chi_\tau})=\rmH_t^K\big\}.
\end{align}

Fix any integer $N\in\bbN$. Consider any test $\tPhi=(\tau,\phi_\tau)$ that satisfies the expected stopping time universality constraint with respect to $N$ and ensures positive mismatch exponent. It follows that
\begin{align}
d(\tilbbP_t(\calW),\bbP_l(\calW))
&\leq D(\tilbbP_t\|\bbP_l)|_{\calF_\tau}\label{dpi:kl}\\
&=\bbE_{\tilbbP_t}\bigg[\sum_{i\in[M_1]}\sum_{n\in[\xi_\tau]}\log\frac{\tilP_i(X_n)}{P_i(X_n)}+\sum_{j\in[M_2]}\sum_{n\in[\chi_\tau]}\log\frac{\tilQ_j(Y_n)}{Q_j(Y_n)}\bigg]\\
&\leq \sum_{i\in[M_1]}(\alpha\bbE_{\tilbbP_t}[\tau]+1)D(\tilP_i\|P_i)+\sum_{j\in[M_2]}(\beta\bbE_{\tilbbP_t}[\tau]+1)D(\tilQ_j\|Q_j)\label{doob}\\
&\leq \sum_{i\in[M_1]}(N\alpha+1)D(\tilP_i\|P_i)+\sum_{j\in[M_2]}(N\beta+1)D(\tilQ_j\|Q_j),\label{converse:step1}
\end{align}
where \eqref{dpi:kl} follows from the data processing inequality of KL divergence~\cite[Theorem 2.8.1]{cover2012elements}, and \eqref{doob} follows from Doob's optimal stopping theorem~\cite{klenke2014optional}.

Note that
\begin{align}
\tilbbP_t(\calW)
&=\tilbbP_t\{\tPhi(\bX^{\xi_\tau},\bY^{\chi_\tau})=\rmH_t^K\}\\
&=1-\beta(\tPhi|\tilP^{M_1},\tilQ^{M_2}),\\
\bbP_l(\calW)
&=\bbP_l\{\tPhi(\bX^{\xi_\tau},\bY^{\chi_\tau})=\rmH_t^K\}\\
&\leq \bbP_l\{\tPhi(\bX^{\xi_\tau},\bY^{\chi_\tau})\neq \rmH_l^K\}\\
&=\beta(\tPhi|P^{M_1},Q^{M_2}).
\end{align}
Since positive mismatch exponent is ensured by $\tPhi$, as $N\to\infty$, it follows that $\beta(\tPhi|\tilP^{M_1},\tilQ^{M_2})\to 0$ and $\beta(\tPhi|P^{M_1},Q^{M_2})\to 0$. Thus, $\tilbbP_t(\calW)\to 1$ and $\bbP_l(\calW)\to 0$. It follows that
\begin{align}
d(\tilbbP_t(\calW),\bbP_l(\calW))
&\geq d(1-\beta(\tPhi|\tilP^{M_1},\tilQ^{M_2}),\beta(\tPhi|P^{M_1},Q^{M_2}))\label{use:bKL}\\
&\geq -\log \beta(\tPhi|P^{M_1},Q^{M_2})\label{converse:step2},
\end{align}
where \eqref{use:bKL} follows since the binary KL divergence decreases in $q$ when $p>q$.

Combining \eqref{converse:step1} and \eqref{converse:step2} leads to
\begin{align}
\limsup_{N\to\infty}\frac{-\log \beta(\tPhi|P^{M_1},Q^{M_2})}{N}
&\leq  \sum_{i\in[M_1]}\alpha D(\tilP_i\|P_i)+\sum_{j\in[M_2]}\beta D(\tilQ_j\|Q_j)\label{converse:step3}.
\end{align}
Note that \eqref{converse:step3} holds for any $(\tilP^{M_1},\tilQ^{M_2})\in\calP_t^K$ and any $t\in[T_K]$ such that $t\neq l$. Thus, to obtain a tight upper bound, one needs to minimize the right hand side of \eqref{converse:step3}, which yields
\begin{align}
\nn&\min_{t\in[T_K]:~t\neq l}\quad \min_{(\tilP^{M_1},\tilQ^{M_2})\in\calP_l^K}\Big(\sum_{i\in[M_1]}\alpha D(\tilP_i\|P_i)+\sum_{j\in[M_2]}\beta D(\tilQ_j\|Q_j)\Big)\\*
&=\min_{t\in[T_K]:~t\neq l}\sum_{(i,j)\in\calM_t^K\setminus\calM_l^K}\alpha D_{\frac{\beta}{\alpha+\beta}}(Q_j\|P_i)\label{useach},
\end{align}
where \eqref{useach} follows from exactly the same steps leading to \eqref{explain4}.

The converse proof of Theorem \ref{result:known} is now completed.

\section{Proof for the Case of Unknown Number of Matches (Theorem \ref{result:unknown})}
\label{proof:unknown}
\subsection{Analysis under Null Hypothesis}
First consider the null hypothesis $\rmH_\rmr$. Fix any $(P^{M_1},Q^{M_2})\in\calP_0$. Recall that $\bbP_\rmr$ is the joint distribution of all sequences of two databases under hypothesis $\rmH_\rmr$.

\subsubsection{Expected Stopping Time}
It follows from our sequential test design that
\begin{align}
\bbE[\tau]
&=\sum_{n\in\bbN}\bbP_\rmr\{\tau>n\}\\
&=N-1+\sum_{n\in\bbN_{N-1}}\bbP_\rmr\{\tau>n\}\label{etau:hr:step1}.
\end{align}

Fix any $n\in\bbN_{N-1}$. Given any integers $(n_1,n_2,n_3)\in\bbN^3$, define the function
\begin{align}
g(n_1,n_2,n_3)&:=\frac{n_2|\calX|\log(n_1\alpha+2)+n_3|\calX|\log (n_1\beta+2)}{n}\label{def:gn}.
\end{align}
For simplicity, we use $g(n)$ to denote $g(n,M_1,M_2)$. Note that $g(n)$ decreases in $n$ and $\lim_{n\to\infty}g(n)=0$. When $N$ is sufficiently large, the second-term term in \eqref{etau:hr:step1} satisfies
\begin{align}
\sum_{n\in\bbN_{N-1}}\bbP_\rmr\{\tau>n\}
&=\sum_{n\in\bbN_{N-1}}\bbP_\rmr\Big\{(\calA^n)^\rmc\cap(\calB^n)^\rmc\Big\}\\
&\leq \sum_{n\in\bbN_{N-1}}\bbP_\rmr\big\{(\calA^n)^\rmc\big\}\label{etau:hr:step1.0}\\
&=\sum_{n\in\bbN_{N-1}}\bbP_\rmr\Big\{\exists~(h,t)\in\calM,~\rmS_t^h(\bX^{\xi_n},\bY^{\chi_n})\leq \lambda_1\Big\}\\
&\leq \sum_{n\in\bbN_{N-1}}\sum_{(h,t)\in\calM}\bbP_\rmr\Big\{\rmS_t^h(\bX^{\xi_n},\bY^{\chi_n})\leq \lambda_1\Big\}\\
&\leq \sum_{n\in\bbN_{N-1}}T (\xi_n+1)^{M_1|\calX|}(\chi_n+1)^{M_2|\calX|}\exp(-nE_\rmr(\lambda_1,P^{M_1},Q^{M_2}))
\label{etau:hr:step2},\\
&\leq T\sum_{n\in\bbN_{N-1}}\exp\Big(-n\big(E_\rmr(\lambda_1,P^{M_1},Q^{M_2})-g(n)\big)\Big)\label{usegn}\\
&\leq T\sum_{n\in\bbN_{N-1}}\exp\Big(-n\big(E_\rmr(\lambda_1,P^{M_1},Q^{M_2})-g(N-1)\big)\Big)\label{usegn:decrease}\\
&=T\frac{\exp\big(-(N-1)(E_\rmr(\lambda_1,P^{M_1},Q^{M_2})-g(N-1)\big)}{1-\exp(-(E_\rmr(\lambda_1,P^{M_1},Q^{M_2})-g(N-1))}\label{etau:hr:step3},
\end{align}
where \eqref{etau:hr:step2} follows similarly to \eqref{useDelta} and uses the definition of $E_\rmr(\cdot)$ in \eqref{def:er}, \eqref{usegn} follows from the upper bound on the number of types and the definition of $g(n)$ in \eqref{def:gn}, \eqref{usegn:decrease} follows since $g(n)$ decreases in $n$ when $n$ is sufficiently large and \eqref{etau:hr:step3} follows from the sum of geometric series. Thus, it follows from \eqref{etau:hr:step1} and \eqref{etau:hr:step3} that for $N$ sufficiently large, given any $(P^{M_1},Q^{M_2})\in\calP_0$, $\bbE_{\bbP_\rmr}[\tau]\leq N$ when $\lambda_1<\rmG_0(P^{M_1},Q^{M_2},\alpha,\beta)$ (cf. \eqref{def:G0}).

\subsubsection{False Alarm Exponent}
Next we bound the false alarm exponent. It follows from the definition of the false alarm probability in \eqref{def:etar} and our sequential test design that
\begin{align}
\eta(\Phi|P^{M_1},Q^{M_2})
&=\bbP_\rmr\{\phi_\tau(\bX^{\xi_\tau},\bY^{\chi_\tau})\neq \rmH_\rmr\}\\
&=\sum_{n\in\bbN}\bbP_\rmr\{\tau=n,~(\calA^n)^\rmc\}\\
&\leq \sum_{n\in\bbN_{N-1}}\bbP_\rmr\{(\calA^n)^\rmc\}\\
&\leq T\frac{\exp\big(-(N-1)(E_\rmr(\lambda_1,P^{M_1},Q^{M_2})-g(N-1)\big)}{1-\exp(-(E_\rmr(\lambda_1,P^{M_1},Q^{M_2})-g(N-1))}\label{use:etau:hr:step3},
\end{align}
where \eqref{use:etau:hr:step3} follows from the results from \eqref{etau:hr:step1.0} to \eqref{etau:hr:step3}. Thus, the false alarm exponent satisfies
\begin{align}
\liminf_{N\to\infty}\frac{-\log \eta(\Phi|P^{M_1},Q^{M_2})}{N}\geq E_\rmr(\lambda_1,P^{M_1},Q^{M_2}).
\end{align}

\subsection{Analyses Under Non-Null Hypotheses}
Fix any $K\in[M_2]$, $l\in[T_K]$ and $(P^{M_1},Q^{M_2})\in\calP_l^K$. Recall that $\bbP_l^K$ is the joint distribution of all sequences of the two databases under hypothesis $\rmH_l^K$.
\subsubsection{Expected Stopping Time}
Under hypothesis $\rmH_l^K$, similarly to \eqref{etau:hr:step1}, it follows from our sequential test design that
\begin{align}
\bbE_{\bbP_l^K}[\tau]
&=N-1+\sum_{n\in\bbN_{N-1}}\bbP_\rmr\{\tau>n\}\label{etau:hlk:step1}.
\end{align}
For each $n\in\bbN_{N-1}$ and when $K<M_2$, it follows from the definition of our test in Section \ref{sec:seq_test:uk} that each probability term in \eqref{etau:hlk:step1} satisfies
\begin{align}
\bbP_l^K\{\tau>n\}
&=\bbP_l^K\Big\{(\calA^n)^\rmc\cap(\calB^n)^\rmc\Big\}\\
&\leq  \bbP_l^K\big\{(\calB^n)^\rmc\big\}\\
&\leq \bbP_l^K\bigg\{\Big(\calB_{K,l}^n\bigcap_{(h,t)\in\calM_{\setminus l}^{\setminus K}}(\calB_{h,t}^n)^\rmc\Big)^\rmc\bigg\}\\
&\leq \bbP_l^K\big\{(\calB_{K,l}^n)^\rmc\big\}+\sum_{(h,t)\in\calM_{\setminus l}^{\setminus K}}\bbP_l^K\big\{\calB_{h,t}^n\big\}\label{etau:hlk:step1:0}.
\end{align}

It follows from the definition of $\calB_{K,l}^n$ in \eqref{def:calBht} that the first term in \eqref{etau:hlk:step1:0} can be upper bounded as follows:
\begin{align}
\bbP_l^K\big\{(\calB_{K,l}^n)^\rmc\big\}
&\leq \bbP_l^K\big\{\rmS_l^K(\bX^{\xi_n},\bY^{\chi_n})>\lambda_2\big\}+\bbP_l^K\Big\{\min_{t\in[T_K]:~t\neq l}\rmS_t^K(\bX^{\xi_n},\bY^{\chi_n})\leq \lambda_3\Big\}\label{etau:hlk:step2:0}\\
\nn&\leq \bbP_l^K\big\{\rmS_l^K(\bX^{\xi_n},\bY^{\chi_n})>\lambda_2\big\}\\*
&\qquad+\bbP_l^K\Big\{\exists~(t_1,t_2)\in[T_K]^2:~t_1\neq t_2,~\rmS_{t_1}^K(\bX^{\xi_n},\bY^{\chi_n})\leq \lambda_3,~\rmS_{t_2}^K(\bX^{\xi_n},\bY^{\chi_n})\leq \lambda_3\Big\}\label{etau:hlk:step2}.
\end{align}
Recall the definition of $g(\cdot)$ in \eqref{def:gn}. Using \eqref{etau:step1.0} to \eqref{etau:step1.1} with $f(n)$ replaced by $\lambda_1$, the first term in \eqref{etau:hlk:step2} is upper bounded by
\begin{align}
\bbP_l^K\big\{\rmS_l^K(\bX^{\xi_n},\bY^{\chi_n})>\lambda_2\big\}
&\leq (n\alpha+2)^{K|\calX|}(n\beta+2)^{K|\calX|}\exp(-n\lambda_2)\\
&\leq \exp\big(-n(\lambda_2-g(n,K,K))\big)\label{use:gn:hlk}.
\end{align}
Recall the definitions of $F(\cdot)$ in \eqref{def:exponent:nonnull} and $g(\cdot)$ in \eqref{def:gn}. Fix any $K\in[M_2]$, $l\in[T_K]$ any $(P^{M_1},Q^{M_2})\in\calP_l^K$. 
Following the steps leading to \eqref{useDelta}, the second term in \eqref{etau:hlk:step2} is upper bounded by
\begin{align}
\nn&\bbP_l^K\Big\{\exists~(t_1,t_2)\in[T_K]^2:~t_1\neq t_2,~\rmS_{t_1}^K(\bX^{\xi_n},\bY^{\chi_n})\leq \lambda_3,~\rmS_{t_2}^K(\bX^{\xi_n},\bY^{\chi_n})\leq \lambda_3\Big\}\\*
&\leq T_K(T_K-1)\exp\Big(-n\big(F(\lambda_3,P^{M_1},Q^{M_2})-g(n,M_1,M_2)\big)\Big)\label{use:eth}.
\end{align}

We next analyze the second term in \eqref{etau:hlk:step1:0}, where the probability term is analyzed for three different cases depending the value of $h$. Recall the definition of $\calB_{2,h,t}$ in \eqref{def:calBht}.
\begin{itemize}
\item Consider any $(h,t)\in\calM_{\setminus l}^{\setminus K}$ such that $h=K$. Similarly to \eqref{use:gn:hlk}, it follows that
\begin{align}
\bbP_l^K\Big\{\calB_{K,t}^n\Big\}
&\leq \bbP_l^K\bigg\{\min_{\bart\in[T_h]:~\bart\neq t}\rmS_{\bart}^K(\bX^{\xi_n},\bY^{\chi_n})>\lambda_3\bigg\}\\
&\leq \bbP_l^K\Big\{\rmS_l^K(\bX^{\xi_n},\bY^{\chi_n})>\lambda_3\Big\}\\
&\leq \exp\Big(-n\big(\lambda_3-g(n,K,K)\big)\Big)\label{p:mis:step2}.
\end{align}

\item Consider any $(h,t)\in\calM_{\setminus l}^{\setminus K}$ such that $h<K$. Note that in this case, one can find $(t_1,t_2)\in[T_h]^2$ such that $t_1\neq t_2$, $\calM_{t_1}^h\subset\calM_l^K$ and $\calM_{t_1}^h\subset\calM_l^K$. Thus, there exists $\bart\in[T_h]$ such that $\calM_{t_1}^h\subset\calM_l^K$. Similarly to \eqref{use:gn:hlk}, it follows that
\begin{align}
\bbP_l^K\Big\{\calB_{h,t}^n\Big\}
&\leq \bbP_l^K\bigg\{\min_{\bart\in[T_h]:~\bart\neq t}\rmS_{\bart}^h(\bX^{\xi_n},\bY^{\chi_n})>\lambda_3\bigg\}\\
&\leq \bbP_l^K\Big\{\rmS_{\bart}^h(\bX^{\xi_n},\bY^{\chi_n})>\lambda_3\Big\}\\
&\leq \exp\Big(-n\big(\lambda_3-g(n,h,h)\big)\Big)\label{p:mis:step3}.
\end{align}

\item Consider any $(h,t)\in\calM_{\setminus l}^{\setminus K}$ such that $h>K$. Recall the definition of $G(\lambda,P^{M_1},Q^{M_2})$ in \eqref{def:g:exponent}. Similarly to \eqref{useDelta}, it follows that
\begin{align}
\bbP_l^K\Big\{\calB_{h,t}^n\Big\}
&\leq \bbP_l^K\Big\{\rmS_t^h(\bX^{\xi_N},\bY^{\chi_N})\le \lambda_2\Big\}\\
&\leq T\exp\Big(-n\big(G(\lambda_2,P^{M_1},Q^{M_2})-g(n,h,h)\big)\Big)\label{p:mis:step4}.
\end{align}
\end{itemize}

Similarly to the steps leading to \eqref{usegn} to \eqref{etau:hr:step3}, for $N$ sufficiently large, it follows from \eqref{etau:hlk:step1}, \eqref{etau:hlk:step1:0}, \eqref{etau:hlk:step2}, \eqref{use:gn:hlk}, \eqref{use:eth}, \eqref{p:mis:step2}, \eqref{p:mis:step3}, \eqref{p:mis:step4} that
\begin{align}
\bbE_{\bbP_l^K}[\tau]
&\leq N-1+\sum_{n\in\bbN_{N-1}}\bbP_l^K\{\tau>n\}\\
\nn&\leq N-1+\frac{\exp\big(-(N-1)(\lambda_2-g(N-1,K,K))\big)}{1-\exp(-(\lambda_2-g(N-1,K,K)))}\\*
\nn&\qquad+T_K(T_K-1)\frac{\exp\big(-(N-1)(F(\lambda_3,P^{M_1},Q^{M_2})-g(N-1,M_1,M_2))\big)}{1-\exp(-(F(\lambda_3,P^{M_1},Q^{M_2})-g(N-1,M_1,M_2))}\\
\nn&\qquad+\sum_{(h,t)\in\calM:~h=K,~t\neq l}\frac{\exp\big(-(N-1)(\lambda_3-g(N-1,M_1,M_2))\big)}{1-\exp(-(\lambda_3-g(N-1,K,K)))}\\
\nn&\qquad+\sum_{(h,t)\in\calM:~h<K}\frac{\exp\big(-(N-1)(\lambda_3-g(N-1,h,h))\big)}{1-\exp(-(\lambda_3-g(N-1,h,h)))}\\
&\qquad+\sum_{(h,t)\in\calM:~h>K}\frac{\exp\big(-(N-1)(G(\lambda_2,P^{M_1},Q^{M_2})-g(N-1,h,h))\big)}{1-\exp(-(G(\lambda_2,P^{M_1},Q^{M_2})-g(N-1,h,h))}
\label{etau:hlk:step3}.
\end{align}
Thus, for $N$ sufficiently large, given any $(P^{M_1},Q^{M_2})\in\calP_l^K$, it follows from Lemma \ref{prop:exponents} that $\bbE_{\bbP_l^K}[\tau]\leq N$  if $\lambda_2<\kappa_l^K(P^{M_1},Q^{M_2},\alpha,\beta)$ (cf. \eqref{def:kappalk}) and $\lambda_3<\Lambda_l^K(P^{M_1},Q^{M_2},\alpha,\beta)$ (cf. \eqref{def:Lambdalk}).

\subsubsection{Mismatch Exponent}
Next we bound mismatch exponent. It follows from the definition of the mismatch probability in \eqref{def:mismatch:unknown} and our sequential test design in Section \ref{sec:seq_test:uk} that
\begin{align}
\bar{\beta}(\Phi|P^{M_1},Q^{M_2})
&=\bbP_l^K\big\{\phi_\tau(\bX^{\xi_\tau},\bY^{\chi_\tau})\notin\{\rmH_l^K,\rmH_\rmr\}\big\}\\
&\leq \sum_{n\in\bbN_{N-1}}\bbP_l^K\Big\{\tau=n,~\exists~(h,t)\in\calM_{\setminus l}^{\setminus K}:~\calB_{h,t}^n\bigcap_{(\barh,\bart)\in\calM_{\setminus t}^{\setminus h}}(\calB_{\barh,\bart}^n)^\rmc\Big\}\\
&\leq \sum_{n\in\bbN_{N-1}}\bbP_l^K\bigg\{\exists~(h,t)\in\calM_{\setminus l}^{\setminus K}:\calB_{h,t}^n\bigg\}\\
&\leq \sum_{n\in\bbN_{N-1}}\sum_{(h,t)\in\calM_{\setminus l}^{\setminus K}}\bbP_l^K\big\{\calB_{h,t}^n\big\}\\
&\leq\sum_{n\in\bbN_{N-1}}\sum_{(h,t)\in\calM_{\setminus l}^{\setminus K}}\bbP_l^K\big\{\calB_{2,h,t}^n\big\}\label{p:mis:step5},
\end{align}
where \eqref{p:mis:step5} follows from the definition of $\calB_{h,t}^n$ in \eqref{def:calBht}.

Similarly to the steps leading to \eqref{usegn} to \eqref{etau:hr:step3}, for $N$ sufficiently large, combining \eqref{p:mis:step2}, \eqref{p:mis:step3}, \eqref{p:mis:step4} and \eqref{p:mis:step5} leads to 
\begin{align}
\liminf_{N\to\infty}\frac{-\log \bar{\beta}(\Phi|P^{M_1},Q^{M_2})}{N}
&\geq \min\big\{G(\lambda_2,P^{M_1},Q^{M_2}),\lambda_3\big\}.
\end{align}
\subsubsection{False Reject Exponent}

Finally, we bound the false reject exponent. It follows from the definition of the false reject probability in \eqref{def:freject} and our sequential test design that
\begin{align}
\zeta(\Phi|P^{M_1},Q^{M_2})
&=\bbP_l^K\big\{\phi_\tau(\bX^{\xi_\tau},\bY^{\chi_\tau})=\rmH_\rmr\big\}\\
&=\sum_{n\in\bbN_{N-1}}\bbP_l^K\big\{\tau=n,~\calA^n\big\}\\
&\leq \sum_{n\in\bbN_{N-1}}\bbP_l^K\big\{\forall~(h,t)\in\calM,~\rmS_t^h(\bX^{\xi_n},\bY^{\chi_n})>\lambda_1\big\}\\
&\leq \sum_{n\in\bbN_{N-1}}\bbP_l^K\big\{\rmS_l^K(\bX^{\xi_n},\bY^{\chi_n})>\lambda_1\big\}\label{p:freject:step1}\\
&\leq \sum_{n\in\bbN_{N-1}}\exp\big(-n(\lambda_1-g(n,K,K))\big)\label{simi:mismatch},
\end{align}
where \eqref{simi:mismatch} follows similarly to \eqref{use:gn:hlk} except that $\lambda_2$ is replaced by $\lambda_1$. Thus, the false reject exponent satisfies
\begin{align}
\liminf_{N\to\infty}\frac{-\log\zeta(\Phi|P^{M_1},Q^{M_2})}{N}\geq \lambda_1.
\end{align}

\section{Conclusion}
\label{sec:conclusion}
We revisited statistical sequence matching, derived large deviations for sequential tests that have bounded expected stopping times and demonstrated the benefit of sequentiality. When the number of matches is known, our results are tight, characterizing the exact mismatch exponent of optimal sequential tests. When the number of matches is unknown, we proposed a non-parametric test, and characterized the tradeoff among exponents of three error probabilities. When specialized to statistical classification, our results strengthened previous studies on sequential tests by allowing the testing sequence to be generated from a distribution that is different from the generating distribution of any training sequence.

There are several avenues for future studies. Firstly, for the case of unknown number of matches, we only derived an achievability result. Without a matching converse result, it is unclear whether our test is optimal or not. Thus, it is worthwhile to derive a converse result for this setting. Secondly, we assumed that each sequence is discrete and extensively applied the method of types to derive large deviations results. However, in practice, the data collected could be continuous. To make a further step towards practical de-anonymization tasks, it is valuable to generalize our results to account for continuous observed sequences, potentially using the kernel method~\cite{gretton2012jmlr}. Thirdly, we assumed that each sequence is generated from an unknown distribution and a pair of sequences is said matched only if they are generated from the same distribution. In practice, even the matched pair of sequences might be generated from distributions that deviate slightly. To account for this case, it is rewarding to generalize our results to account for distribution uncertainty and characterize the impact of the uncertainty level on the performance of optimal tests~\cite{hsu2020binary,pan2022tit}. Finally, we focused on the asymptotical large deviations setting where the sample size tends to infinity. However, any practical problem provides only sequences of finite sample size. It is thus beneficial to study the non-asymptotic performance of optimal sequential tests, potentially extending the ideas in \cite{li2020second}.

\appendix 

\subsection{Achievability Proof of the Fixed-Length Test (Theorem \ref{theo:fl})}
\label{proof:fl:test}

Recall the definition of $g(\cdot)$ in \eqref{def:gn}. Fix any $l\in[T_K]$ and any tuple of distributions $(P^{M_1},Q^{M_2})\in\calP_l^K$. Recall that $\bbP_l^K$ denotes the joint distribution of sequences $(\bX^{\xi_N},\bY^{\chi_N})$. Define the set $[T_K]_{\setminus l}:=\{t\in[T_K]:~t\neq l\}$. Similarly to \eqref{use:prob:typeclass} to \eqref{use:num:types}, it follows from the definition of the mismatch probability in \eqref{def:mismatch} and the test design in \eqref{test:fl} that
\begin{align}
\nn&\beta(\Phi|P^{M_1},Q^{M_2})\\*
&=\bbP_l^K\Big\{\exists~t\in[T_K]_{\setminus l},~\rmS_t^K(\bX^{\xi_N},\bY^{\chi_N})\leq \rmS_l^K(\bX^{\xi_N},\bY^{\chi_N})\Big\}\\
&\leq \sum_{t\in[T_K]_{\setminus l}}\bbP_l^K\big\{\rmS_t^K(\bX^{\xi_N},\bY^{\chi_N})\leq \rmS_l^K(\bX^{\xi_N},\bY^{\chi_N})\big\}\\
&=\sum_{t\in[T_K]_{\setminus l}}\sum_{\substack{(\bx^{\xi_N},\by^{\chi_N}):\\\rmS_t^K(\bx^{\xi_N},\by^{\chi_N})\leq \rmS_l^K(\bx^{\xi_N},\by^{\chi_N})}}\Big(\prod_{i\in[M_1]}P_i^{\xi_N}(x_i^{\xi_N})\Big)\Big(\prod_{j\in[M_2]}Q_j^{\chi_N}(y_j^{\chi_N})\Big)\\
&=\sum_{t\in[T_K]_{\setminus l}}\sum_{\substack{(\Omega^{M_1},\Psi^{M_2})\in(\calP^{\xi_n}(\calX))^{M_1}\times(\calP^{\chi_n}(\calX))^{M_2}:\\ \rmG_t^K(\Omega^{M_1},\Psi^{M_2},\alpha,\beta)\leq \rmG_l^K(\Omega^{M_1},\Psi^{M_2},\alpha,\beta)
}}\Big(\prod_{i\in[M_1]}P_i^{\xi_n}(\calT_{\Omega_i}^{\xi_n})\Big)\Big(\prod_{j\in[M_2]}Q_j^{\chi_n}(\calT_{\Psi_j}^{\chi_n})\Big)\\
&\leq \sum_{t\in[T_K]_{\setminus l}}\sum_{\substack{(\Omega^{M_1},\Psi^{M_2})\in(\calP^{\xi_n}(\calX))^{M_1}\times(\calP^{\chi_n}(\calX))^{M_2}:\\ \rmG_t^K(\Omega^{M_1},\Psi^{M_2},\alpha,\beta)\leq \rmG_l^K(\Omega^{M_1},\Psi^{M_2},\alpha,\beta)
}}\exp\bigg(-\Big(\sum_{i\in[M_1]}\xi_nD(\Omega_i\|P_i)+\sum_{j\in[M_2]}\chi_nD(\Psi_j\|Q_j)\Big)\bigg)\\
&\leq \sum_{t\in[T_K]_{\setminus l}}\sum_{\substack{(\Omega^{M_1},\Psi^{M_2})\in(\calP^{\xi_n}(\calX))^{M_1}\times(\calP^{\chi_n}(\calX))^{M_2}:\\ \rmG_t^K(\Omega^{M_1},\Psi^{M_2},\alpha,\beta)\leq \rmG_l^K(\Omega^{M_1},\Psi^{M_2},\alpha,\beta)
}}\exp\bigg(-n\Big(\sum_{i\in[M_1]}\alpha D(\Omega_i\|P_i)+\sum_{j\in[M_2]}\beta D(\Psi_j\|Q_j)\Big)\bigg)\\
&\leq \sum_{t\in[T_K]_{\setminus l}}\sum_{\substack{(\Omega^{M_1},\Psi^{M_2})\in(\calP(\calX))^{M_1+M_2}:\\ \rmG_t^K(\Omega^{M_1},\Psi^{M_2},\alpha,\beta)\leq \rmG_l^K(\Omega^{M_1},\Psi^{M_2},\alpha,\beta)
}}\exp\big(-n E(P^{M_1},Q^{M_2},\Omega^{M_1},\Psi^{M_2})\big)\label{useegain}\\
&\leq \sum_{t\in[T_K]_{\setminus l}}\sum_{\substack{(\Omega^{M_1},\Psi^{M_2})\in(\calP(\calX))^{M_1+M_2}:\\ \rmG_t^K(\Omega^{M_1},\Psi^{M_2},\alpha,\beta)\leq \rmG_l^K(\Omega^{M_1},\Psi^{M_2},\alpha,\beta)
}}\exp\big(-n E_\rmf(l,K,P^{M_1},Q^{M_2})\big)\label{use:e:fl}\\
&\leq T_K\exp\Big(-n\Big(E_\rmf(l,K,P^{M_1},Q^{M_2})-g(N,M_1,M_2)\Big)\Big)\label{use:g:fl},
\end{align}
where \eqref{useegain} follows from the definition of $E(\cdot)$ in \eqref{def:epqop}, \eqref{use:e:fl} follows from the definition of $E_\rmf(\cdot)$ in \eqref{def:exponent:fl}, and \eqref{use:g:fl} follows from the upper bound on the number of types and the definition of $g(\cdot)$ in \eqref{def:gn}.

Thus, the mismatch exponent of the fixed-length test satisfies
\begin{align}
\liminf_{N\to\infty}\frac{-\log \beta(\Phi|P^{M_1},Q^{M_2})}{n}\geq E_\rmf(l,K,P^{M_1},Q^{M_2}).
\end{align}

\subsection{Justification of \eqref{eqn:67}}
\label{just:eqn:67}
Consider any $(P^{M_1},Q^{M_2})\in\calP_0$ (cf. \eqref{def:calP:r}). For any $(h,t)\in\calM$, it follows that
\begin{align}
\sum_{(i,j)\in\calM_t^h}\mathrm{GJS}(P_i,Q_j,\alpha,\beta) \geq \min_{(i,j)\in[M_1]\times[M_2]}\mathrm{GJS}(P_i,Q_j,\alpha,\beta).
\end{align}
As a result, 
\begin{align}
\min_{(h,t)\in\calM}\sum_{(i,j)\in\calM_t^h}\mathrm{GJS}(P_i,Q_j,\alpha,\beta) \geq \min_{(i,j)\in[M_1]\times[M_2]}\mathrm{GJS}(P_i,Q_j,\alpha,\beta).  \label{eqn:lowerbou}
\end{align}
On the other hand, consider any pair
\begin{align}
(i_0,j_0)\in \argmin_{(i,j)\in[M_1]\times[M_2]}  \mathrm{GJS}(P_i,Q_j,\alpha,\beta).
\end{align}
Assume that $h_0=1$. Consider hypothesis $\calM_{h_0}^{t_0}$ such that $\calM_{t_0}^{h_0}=\{ (i_0,j_0) \}$. It follows that 
\begin{align}
\min_{(h,t)\in\calM}\sum_{(i,j)\in\calM_t^h}\mathrm{GJS}(P_i,Q_j,\alpha,\beta)
&\leq \sum_{(i,j)\in\calM_{t_0}^{h_0}}\mathrm{GJS}(P_i,Q_j,\alpha,\beta)\\
&= \mathrm{GJS}(P_{i_0},Q_{j_0},\alpha,\beta)\\
&=\min_{(i,j)\in[M_1]\times[M_2]}\mathrm{GJS}(P_i,Q_j,\alpha,\beta)\label{eqn:upperbou}.
\end{align} 
The justification of \eqref{eqn:67} is completed by combining \eqref{eqn:lowerbou} and \eqref{eqn:upperbou}.

\subsection{Achievability Proof for Another Fixed-Length Test (Theorem \ref{theo:fl:one-step})}
\label{proof:fl:one-step}

Recall the definitions of $E_\rmr(\lambda_1,P^{M_1},Q^{M_2})$ in \eqref{def:er} and $g(\cdot)$ in \eqref{def:gn}. Fix any $(P^{M_1},Q^{M_2})\in\calP_0$ (cf. \eqref{def:calP:r}). Recall that $\bbP_\rmr$ denotes the joint distribution of all sequences of two databases. It follows from \eqref{def:etar} that the false alarm probability satisfies
\begin{align}
\eta(\Phi_{\rm{FL}}^{\rm{uk}}|P^{M_1},Q^{M_2})
&=\bbP_\rmr\Big\{\Phi_{\rm{FL}}^{\rm{uk}}(\bX^{\xi_N},\bY^{\chi_N})\neq \rmH_\rmr\Big\}\\
&=\bbP_\rmr\bigg\{\exists~(\barh,\bart)\in\calM,~\calB_{\barh,\bart}^N\bigcap_{(h,t)\in\calM_{\setminus \bart}^{\setminus \barh}}(\calB_{h,t}^n)^\rmc\bigg\}\\
&\leq \sum_{(\barh,\bart)\in\calM}\bbP_\rmr\Big\{\calB_{\barh,\bart}^N\Big\}\\
&\leq \sum_{(\barh,\bart)\in\calM}\bbP_\rmr\Big\{\rmS_{\barh}^{\bart}(\bX^{\xi_N},\bY^{\chi_N})\leq \lambda_1\Big\}\label{use:calA2:omit}\\
&\leq T\exp\Big(-N\big(E_\rmr(\lambda_1,P^{M_1},Q^{M_2})-g(N,M_1,M_2)\big)\Big)\label{use:usegn:fl:uk},
\end{align}
where \eqref{use:calA2:omit} follows from the definition of $\calB_{2,h,t}$ in \eqref{def:calBht} and \eqref{use:usegn:fl:uk} follows from the result \eqref{usegn} by replacing $n$ with $N$ and ignoring the outer summation over $n$. Thus, the false alarm exponent satisfies
\begin{align}
\liminf_{N\to\infty}\frac{-\log \eta(\Phi_{\rm{FL}}^{\rm{uk}}|P^{M_1},Q^{M_2})}{N}\geq E_\rmr(\lambda_1,P^{M_1},Q^{M_2}).
\end{align}

Next consider non-null hypotheses. Fix $(K,l)\in\calM$ and $(P^{M_1},Q^{M_2})\in\calP_l^K$ (cf. \eqref{def:calp:lk}). Recall that $\bbP_l^K$ denotes the joint distribution of all sequences of the two databases. It follows from \eqref{def:mismatch:unknown} that the mismatch probability satisfies
\begin{align}
\beta(\Phi_{\rm{FL}}^{\rm{uk}}|P^{M_1},Q^{M_2})
&=\bbP_l^K\Big\{\Phi_{\rm{FL}}^{\rm{uk}}(\bX^{\xi_N},\bY^{\chi_N})\notin\{\rmH_l^K,\rmH_\rmr\}\Big\}\\
&\leq \bbP_l^K\bigg\{\exists~(\barh,\bart)\in\calM_{\setminus l}^{\setminus K},~\calB_{\barh,\bart}^N\bigcap_{(h,t)\in\calM_{\setminus \bart}^{\setminus \barh}}(\calB_{h,t}^n)^\rmc\bigg\}\\
&\leq \bbP_l^K\Big\{\exists~(\barh,\bart)\in\calM_{\setminus l}^{\setminus K},~\calB_{\barh,\bart}^N\Big\}\\
&\leq \sum_{(\barh,\bart)\in\calM_{\setminus l}^{\setminus K}}\bbP_l^K\{\calB_{\barh,\bart}^N\}\\
&\leq T_K\exp\Big(-N\big(\lambda_2-g(N,M_1,M_2)\big)\Big)+T\exp\Big(-N\big(\lambda_2-g(N,h,h)\big)\Big)\\*
&\qquad+T\exp\Big(-N\big(G(\lambda_1,P^{M_1},Q^{M_2})-g(N,M_1,M_2)\big)\Big)
\label{use2:use:gn:hlk},
\end{align}
where \eqref{use2:use:gn:hlk} follows from the results in \eqref{p:mis:step2}, \eqref{p:mis:step3}, and \eqref{p:mis:step4} with $n$ replaced by $N$. Thus, the mismatch exponent satisfies
\begin{align}
\liminf_{N\to\infty}\frac{-\log \beta(\Phi_{\rm{FL}}^{\rm{uk}}|P^{M_1},Q^{M_2})}{N}\geq \min\Big\{G(\lambda_1,P^{M_1},Q^{M_2}),\lambda_2\Big\}.
\end{align}

It follows from \eqref{def:freject} that the false reject probability satisfies
\begin{align}
\zeta(\Phi_{\rm{FL}}^{\rm{uk}}|P^{M_1},Q^{M_2})
&=\bbP_l^K\Big\{\Phi_{\rm{FL}}^{\rm{uk}}(\bX^{\xi_N},\bY^{\chi_N})=\rmH_\rmr\Big\}\\
&\leq\bbP_l^K\Big\{\big(\calB_{K,l}^N\big)^\rmc\bigcup_{(h,t)\in\calM_{\setminus l}^{\setminus K}}\calB_{h,t}^n\Big\}\\
&\leq \bbP_l^K\big\{(\calB_{K,l}^N)^\rmc\big\}+\sum_{(h,t)\in\calM_{\setminus l}^{\setminus K}}\bbP_l^K\big\{\calB_{h,t}^n\big\}
\label{fl:uk:fr:step1}.
\end{align}

Recall the definition of $F(\lambda_2,P^{M_1},Q^{M_2})$ in \eqref{def:exponent:nonnull}. It follows from \eqref{etau:hlk:step2:0}, \eqref{etau:hlk:step2}, \eqref{use:gn:hlk} and \eqref{use:eth} that the first term in \eqref{fl:uk:fr:step1} satisfies
\begin{align}
\bbP_l^K\big\{(\calB_{K,l}^N)^\rmc\big\}
&\leq \exp\Big(-n\big(\lambda_1-g(n,K,K)\big)\Big)+T_K\exp\Big(-n\big(F(\lambda_2,P^{M_1},Q^{M_2})-g(n,M_1,M_2)\big)\Big)\label{fl:uk:fr:step2}.
\end{align}
Recall the definition of $G(\lambda,P^{M_1},Q^{M_2})$ in \eqref{def:g:exponent}. The second term in \eqref{fl:uk:fr:step1} has been bounded in \eqref{use2:use:gn:hlk}.

Therefore, combining \eqref{use2:use:gn:hlk}, \eqref{fl:uk:fr:step1}, and \eqref{fl:uk:fr:step2}, the false reject exponent satisfies
\begin{align}
\liminf_{N\to\infty}\frac{-\log\zeta(\Phi_{\rm{FL}}^{\rm{uk}}|P^{M_1},Q^{M_2})}{N}
\geq \min\Big\{\lambda_1,\lambda_2,G(\lambda_1,P^{M_1},Q^{M_2}),F(\lambda_2,P^{M_1},Q^{M_2})\Big\}.
\end{align}

\bibliographystyle{IEEEtran}
\bibliography{IEEEfull_lin}
\end{document}